\documentclass{rspublic}
\usepackage{amsfonts,amssymb}
\usepackage{epsfig}
\usepackage{url}

\newtheorem{comment}{Comment}
\begin{document}

\title[ Riemann Hypothesis for Combined Zeta Functions]{ The  Riemann Hypothesis for Symmetrised Combinations of Zeta Functions}
\author[R.C. McPhedran \& C.G. Poulton ]{Ross C. McPhedran$^1$ and Christopher G. Poulton$^2$
}

\affiliation{$^1$ Department of Mathematical Sciences, University of Liverpool,  Liverpool L69 7ZL, United Kingdom, and CUDOS, School of Physics, University of Sydney, NSW 2006, Australia
\\
$^2$School of Mathematical Sciences, University of Technology,Sydney, N.S.W. 2007 Australia\\
}
\label{firstpage}

\maketitle
\begin{abstract}{Lattice sums, Riemann hypothesis, distribution functions of zeros}
This paper studies combinations of the Riemann zeta function, based on one defined by P.R. Taylor, which was shown by him to have all its zeros on the critical line. With a rescaled complex argument, this is denoted here by ${\cal T}_-(s)$, and is considered together with a counterpart function ${\cal T}_+(s)$, symmetric rather than antisymmetric about the critical line. We prove  that ${\cal T}_+(s)$ has all its zeros on the critical line, and that the zeros of both functions are all of first order. We establish a link between  the zeros of ${\cal T}_-(s)$ and of ${\cal T}_+(s)$ with those of the zeros of the Riemann zeta function  $\zeta(2 s-1)$, which enables us to prove  if the Riemann hypothesis holds then the distribution function of the zeros of $\zeta (2 s-1)$ agrees with those for  ${\cal T}_-(s)$ and of ${\cal T}_+(s)$ in all terms which do not remain finite as $t\rightarrow \infty$. \end{abstract}
\section{Introduction}
The Riemann hypothesis is considered to be one of the most important unsolved problems in mathematics, and a central issue in analytic number theory (Titchmarsh and Heath-Brown, 1987), as, if proved, it would give more precise estimates for the distribution of zeros of the Riemann zeta function $\zeta (s)$of the complex argument $s=\sigma+i t$, and in turn more precise information about the distribution of prime numbers.  The hypothesis itself is that all zeros of $\zeta (s)$ have $\sigma=1/2$, and its proof has resisted the efforts of many famous mathematicians. However, their efforts have resulted in an impressive and extensive literature, dealing with many properties of $\zeta (s)$ and closely related functions.

Among many results concerning the Riemann hypothesis which have been proved, we cite the property that all zeros of $\zeta (s)$ lie in the critical strip $0<\sigma<1/2$, and that the number of zeros lying on the critical line is infinite   (Titchmarsh and Heath-Brown, 1987). The Bohr-Landau Theorem (Edwards, 1974) proves that all but an infinitesimal fraction of zeros of $\zeta (s)$ lie arbitrarily close to the critical line. This important result however does not constrain the proportion of zeros actually lying on the critical line. Work on constraining this fraction has continued, with the results improving from 1/3  (Levinson, 1974) to 2/5 (Conrey, 1989) and latterly to 41\%  (Bui, Conrey and Young, 2011).
There of course have been numerical investigations as to the truth of the Riemann hypothesis for finite ranges of $\Im(s)=t$ (see for example Chapter 15 in Titchmarsh and Heath-Brown, 1987), with it being established that the first $1.5\times 10^{9}$ zeros of $\zeta (s)$ are simple and lie on the critical line. (The fact that all zeros of $\zeta (s)$ are simple is an important property of $\zeta(s)$ yet to be proved.)

Such a famous conjecture as the Riemann hypothesis has attracted the attention of physicists, and indeed is of importance to them, since the zeta function occurs in regularisation arguments in quantum field theories (Elizalde, 1995). An excellent review of the links between physics and the Riemann hypothesis has been given by Schumayer and Hutchinson (2011). The connections they mention are in classical mechanics, quantum mechanics, nuclear physics, condensed matter physics, and statistical physics. The path the present authors and their colleagues followed leading to investigations into the hypothesis commenced with a paper by Lord Rayleigh (1892),  which has been applied in
studies of photothermal solar absorbers (McPhedran \& McKenzie, 1978, Perrins, McKenzie \& McPhedran, 1979),  photonic crystals   and metamaterials  (Nicorovici, McPhedran \& Botten, 1995)  . Its underlying mathematics involved the properties of conditionally-convergent double sums, and
the class of such sums used by Rayleigh has interesting connections with the Riemann hypothesis.

%This paper continues  the investigations into the Riemann hypothesis for sums related to the Riemann zeta function. 
To be more specific,  in a recent paper (McPhedran et al 2011, hereafter referred to as I)
it was established that the Riemann hypothesis held for all of a particular class of two-dimensional sums over the square lattice if and only if it held for the lowest member of that class. That lowest
member, denoted ${\cal C}(0,1;s)$ was known from  the work of Lorenz (1871) and Hardy (1920)
to be given by the product of the Riemann zeta function $\zeta (s)$, and a particular Dirichlet $L$ function, $L_{-4}(s)$. In I it was also established that, if  all zeros of ${\cal C}(0,1;s)$ lie on $\sigma=1/2$ by a generalisation of the Riemann hypothesis, then the whole class of sums has the same distribution function for zeros,
as far as all terms tending to infinity with $t$ are concerned. 

Here we wish to link these results to an established result for a combination of Riemann zeta functions given by P.R. Taylor, and published in a posthumous paper in 1945 (Taylor, 1945).
An editor's note to this paper refers to  Flight Lieutenant P.R. Taylor as ``missing, believed killed, on active service in November 1943''. It also notes that Professor Titchmarsh revised and completed
Taylor's argument.  What Taylor proved {\em inter alia} was that $\xi_1(s+1/2)-\xi_1(s-1/2)$
has all its zeros on $\sigma=1/2$, where $\xi_1(s)=\Gamma(s/2) \zeta(s)/\pi^{s/2}$ is a symmetrised form of $\zeta(s)$. After this proposition was stated, it was commented that it``may be a step towards the proof of the Riemann hypothesis''. We will take this comment  further here, in order to establish that Taylor's work does indeed furnish an only partially-explored way of investigating the Riemann hypothesis, together with related questions.

Google Scholar lists ten citations of Taylor's paper. Of the seven relevant to our concerns here, two (Anderson, 1986 and Matsumoto \& Tanigawa) use Taylor's method of proof. The other five have been published in the interval 2006-2010, testifying to recent interest in the properties of functions closely associated with the Riemann zeta function. The survey article of Balazard (2010) refers to Taylor's paper as a "remarkable contribution to the theory of the $\zeta$ function", and to his proof as "classical and elegant". It is also commented that
Velasquez Casta\~{n}on (2010)  extended Taylor's method into a very general context, and showed that numerous recent results flowed naturally from the approach. 

The combination of zeta functions considered by Taylor is antisymmetric under the replacement $s\rightarrow 1-s$. Here, we will complement Taylor's function by a symmetric combination, and replace the variable $s$ by $2 s-1/2$, in  order to provide a link with the two-dimensional sums studied in I. The resulting functions will be denoted ${\cal T}_-(s)$
and  ${\cal T}_+(s)$ respectively. We will give numerical evidence (and later prove) that ${\cal T}_-(s)$ and ${\cal T}_+(s)$ both have the same distribution functions of zeros on  the critical line $\sigma=1/2$, and that this is also the same as that of ${\cal C}(0,1;s)$ and $\zeta (2 s-1/2)$.

We give in Section 2  the definitions of relevant functions, which are ${\cal T}_+(s)$,  ${\cal T}_-(s)$,
their ratio ${\cal V}(s)$, ${\cal U}(s)= ({\cal V}(s)-1)/ ({\cal V}(s)+1)$ and ${\cal C}(0,1;s)$, and provide numerical evidence concerning their distributions of zeros on the critical line. 
The numerical evidence shown motivated the proofs given in Sections 3 and 4, but of course does not in itself constitute an essential element of those proofs.
We discuss in Section 3 the morphologies of relevant functions in the regions around their zeros and poles. This provides the basis for the proofs in Section 4 that ${\cal T}_+(s)$, in addition to 
${\cal T}_-(s)$ as proved by Taylor (1945), has all its zeros on the critical line, and that all zeros of both functions are simple. In Section 4, we also make explicit the links between the zeros and the distribution functions of zeros of ${\cal T}_-(s)$, ${\cal T}_+(s)$ and those of the Riemann zeta function  $\zeta(2 s-1)$. These links in fact enable us to put forward a proof that all the zeros of
 $\zeta(2 s-1)$  have a number distribution function that must agree with those for ${\cal T}_-(s)$ and  ${\cal T}_+(s)$ in all terms which diverge as
 $t\rightarrow \infty$ . An important feature of Section 4 is that the proofs advanced of the mathematical propositions are based on a property of an equivalent electrostatic problem which would be regarded as self-evident by many physicists.
 We include two Appendices, the first of which discusses the convergence of the logarithmic potential function associated with ${\cal U}(s)$.  The second  illustrates modified functions ${\cal U}(s)$ which do not obey the Riemann hypothesis, but from which may be constructed a function ${\cal V}(s)$
 having all the properties listed above (all zeros and poles on the critical line, all these being of first order, no zeros of derivatives on the critical line).

The proofs in this paper  are less formal than is expected for important results in analytic number theory (while being of similar style to that of Speiser,1934), but may have the desirable outcome of motivating the construction of more formal proofs. At all events, their construction is such as to make them accessible and interesting to mathematical physicists
 and applied mathematicians. They also establish that ${\cal T}_-(s)$ and  ${\cal T}_+(s)$ are two functions which provably have  the properties which one would desire to be able to prove for $\zeta(s)$. It may be the case that the comparative method used in I, or a related method, may be used to  forge a link between them and $\zeta(2 s-1/2)$, enabling the resolution of the Riemann hypothesis.
 
After the initial version of this article was completed, the authors became aware of an important earlier paper by Lagarias and Suzuki (2006). There is significant overlap between the 2006 paper and the content of Section 4:  the 2006 paper contains rigorous proofs of propositions for which more physically-based arguments are given here.  Lagarias and Suzuki also give more general forms for combinations of the symmetrized zeta functions $\xi_1(2s)$ and $\xi_1(2 s-1)$ which obey the Riemann hypothesis than those discussed in Section 4. 

\section{Functions and Zero Distributions}
We recall the definition from   McPhedran et al (2011) of two sets of angular
lattice sums for the square array:
\begin{equation}
{\cal C}(n,m;s)=\sum_{p_1,p_2} ' \frac{\cos^n(m
\theta_{p_1,p_2})}{(p_1^2+p_2^2)^s},~~ {\cal
S}(n,m;s)=\sum_{p_1,p_2} ' \frac{\sin^n(m
\theta_{p_1,p_2})}{(p_1^2+p_2^2)^s}, \label{mz-1}
\end{equation}
where $\theta_{p_1,p_2}=\arg (p_1+\ri p_2)$, the prime denotes the
exclusion of the point at the origin, and the complex number $s$ is written 
in terms of real and imaginary parts as $s=\sigma+i t$. The sum independent of the
angle $\theta_{p_1,p_2}$ was evaluated by Lorenz (1871) and Hardy(1920)
in terms of the product of Dirichlet $L$ functions:
\begin{equation}
{\cal C}(0,m;s)={\cal S}(0,m;s)\equiv{\cal C}(0,1;s)\equiv{\cal C}(1,0;s)=4 L_1(s)
L_{-4}(s)=4 \zeta(s) L_{-4}(s) .\label{mz2}
\end{equation}
 A useful
account of the properties of Dirichlet $L$ functions such as $L_{-4}(s)$ has been given
by Zucker \& Robertson (1976).

An  expression for the lowest order sum was derived by  Kober (1936), 
\begin{eqnarray}
{\cal C}(0,1;s) =
2\zeta(2 s)+\frac{2\sqrt{\pi}\Gamma(s-1/2)}{\Gamma(s)}\zeta(2 s-1)
\hspace{25ex} \nonumber \\
+\frac{8\pi^s }{\Gamma(s)}\sum_{p_1=1}^\infty \sum_{p_2=1}^\infty
\left( \frac{p_2}{p_1}\right)^{s-1/2} K_{s-1/2}(2\pi p_1 p_2),
\label{sn3}
\end{eqnarray}
where $K_\nu(z)$ denotes the modified Bessel function of the second
kind, or Macdonald function, with order $\nu$ and argument $z$.
Let us introduce a general notation for  Macdonald function  sums:
\begin{equation}
{\cal K}(n,m;s)= \sum_{p_1,p_2=1}^{\infty} \left(\frac{p_2}{p_1}\right)^{s-1/2} (p_1 p_2\pi)^{ n}   K_{m+s-1/2} (2 \pi p_1 p_2).
\label{n22}
\end{equation}
By interchanging $p_2$ and $p_1$ and using the relation $K_{-\nu}(z)=K_{\nu}(z)$, we see that ${\cal K}(n,0;s)$ is symmetric under the substitution $s\rightarrow 1-s$:
\begin{equation}
{\cal K}(n,0;1-s)= {\cal K}(n,0;s).
\label{n24}
\end{equation}

The lowest symmetric sum is, if $s_-=s-1/2$,
\begin{equation}
{\cal K}(0,0;s)= \frac{\Gamma(s)}{8\pi^s} {\cal C}(0,1;s)-
\left[ \frac{\Gamma(s)\zeta (2s)}{4\pi^s}+\frac{\Gamma(s_-)\zeta (2s_-)}{4\pi^{s_-}}\right].
\label{n25}
\end{equation}
Each of the two terms on the right-hand side of (\ref{n25}) is unchanged by the substitution
$s\rightarrow 1-s$. (That the first term is unchanged follows from the functional equation for 
${\cal C}(0,1;s)$- see  McPhedran et al (2011); that the second is unchanged then follows from the preceding fact and the symmetry of ${\cal K}(0,0;s)$.) They are therefore real on the critical line $\sigma=1/2$.

We now introduce two symmetrised functions, the first of which occurred in equation (\ref{n25}):
\begin{eqnarray}
{\cal T}_+(s)&=&\frac{\Gamma(s)\zeta (2s)}{4\pi^s}+\frac{\Gamma(s_-)\zeta (2s_-)}{4\pi^{s_-}},
\nonumber \\
  &=& \frac{1}{4}\left[  \xi_1( 2 s)+\xi_1(2 s-1) \right],
\label{n26a}
\end{eqnarray}
so that
\begin{equation}
{\cal K}(0,0;s)=\frac{\Gamma(s)}{8\pi^s} {\cal C}(0,1;s)-{\cal T}_+(s).
\label{n25a}
\end{equation}
Its antisymmetric counterpart is
\begin{eqnarray}
{\cal T}_-(s)&=&\frac{\Gamma(s)\zeta (2s)}{4\pi^s}-\frac{\Gamma(s_-)\zeta (2s_-)}{4\pi^{s_-}}
\nonumber \\
& =&  \frac{1}{4}\left[  \xi_1( 2 s)-\xi_1(2 s-1) \right] .
\label{n27}
\end{eqnarray}
This is equivalent to the function considered by P.R. Taylor (1945)
after his variable $s$ is replaced by our $2 s-1/2$. Taylor proved in fact that his function obeys the Riemann hypothesis, as must then ${\cal T}_-(s)$.

Now, from Titchmarsh \& Heath-Brown (1987), the Riemann zeta function has its zeros confined to the region
$0<\sigma<1$, from which we immediately see that ${\cal T}_+(s)$ and ${\cal T}_-(s)$ cannot have
coincident zeros.%%%%%%%%%%%%%%%%%%%%%%%%%%%%%%%%%%%%%%%%%%%%%%%%%%%%%%%%%%%%%%%%%%%%%
\begin{figure}[h]
\includegraphics[width=3.0in]{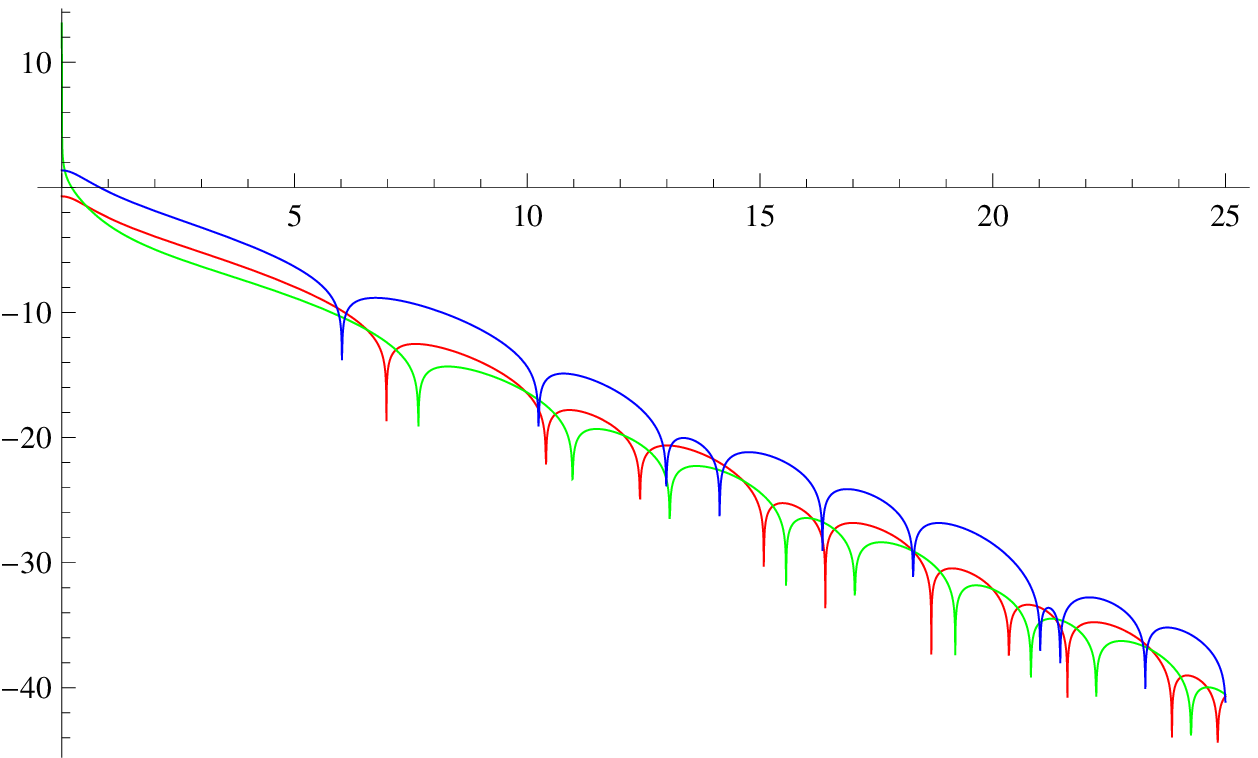}~~\includegraphics[width=3.0in]{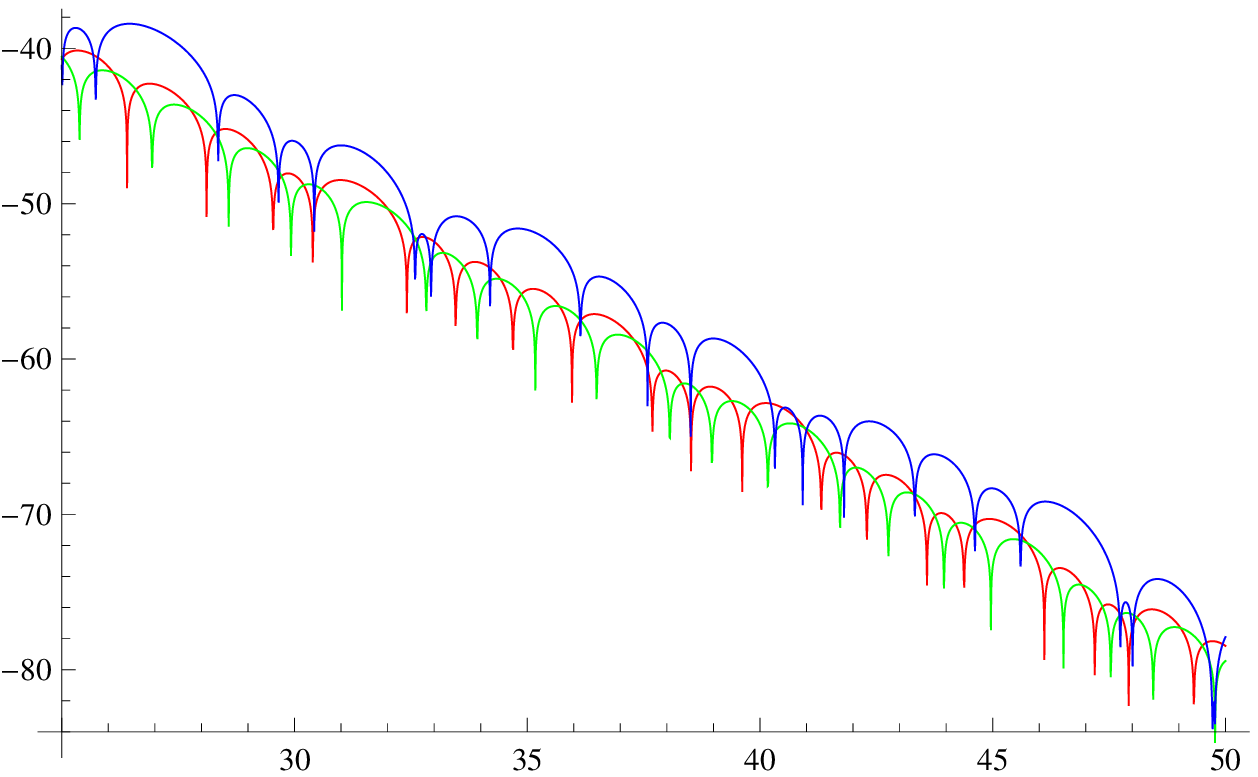}
\caption{\label{fig01} Plots of the logarithmic modulus of the functions ${\cal T}_+(s)$ (red),
${\cal T}_-(s)$ (green)  and ${\cal C}(0,1;s)*\Gamma(s)/\pi^s$ (blue) along the critical line.}\end{figure}
%%%%%%%%%%%%%%%%%%%%%%%%%%%%%%%%%%%%%%%%%%%%%%%%%%%%%%%%%%%%%%%%%%%%%
In Fig. \ref{fig01} we show the variation along the critical line of the functions ${\cal T}_+(s)$ ,
${\cal T}_-(s)$  and ${\cal C}(0,1;s)*\Gamma(s)/\pi^s$. Each has a similar number of zeros in the interval shown.
In Fig.\ref{fig02}  we further show the behaviour on the critical line of the functions ${\cal K}(0,0;s)$ ,
${\cal K}(1,1;s)$  and ${\cal C}(0,1;s)*\Gamma(s)/\pi^s$. It can be seen that the first of these has significantly fewer zeros on the critical line than the third, while the second appears to have none.
In McPhedran et al (2004) it is commented that  ${\cal K}(0,0;s)$  has approximately the same number of zeros (in complex conjugate pairs) off the critical line as on it.
% (see Fig. \ref{fig02} (right)).
%%%%%%%%%%%%%%%%%%%%%%%%%%%%%%%%%%%%%%%%%%%%%%%%%%%%%%%%%%%%%%%%%%%%%
\begin{figure}[h]
\includegraphics[width=5.0in]{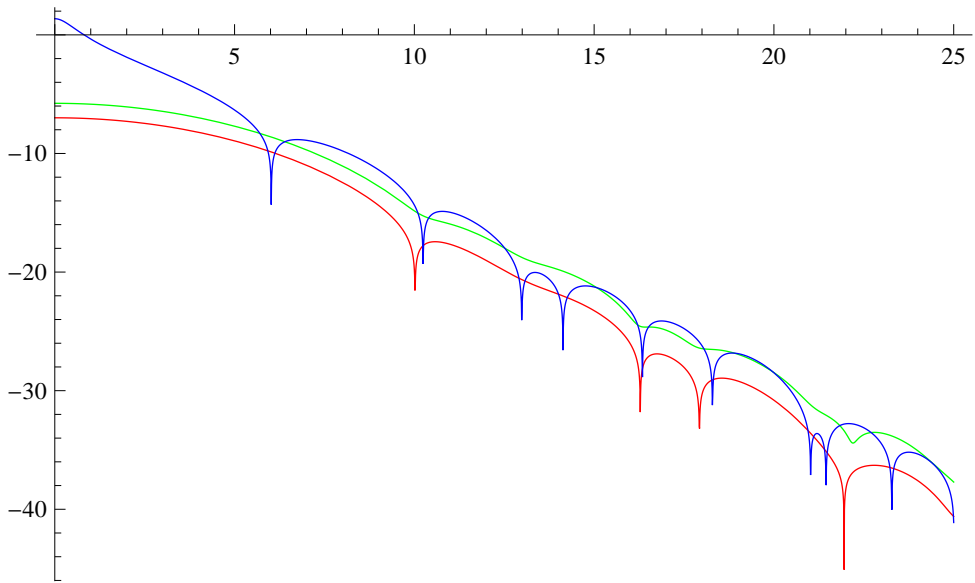}
%~~\includegraphics[width=3.0in]{jmp04fig.pdf}
%\includegraphics[width=3.0in]{transpcfdc.pdf}
\caption{\label{fig02}  Plots of the logarithmic modulus of the functions ${\cal K}(0,0;s)$ (red),
${\cal K}(1,1;s)$ (green)  and ${\cal C}(0,1;s)*\Gamma(s)/\pi^s$ (blue) along the critical line.
%Right: The lowest zeros of  ${\cal K}(0,0;s)$ in the complex plane of $s_-$
% (from McPhedran et al, 2004).
 }\end{figure}
%%%%%%%%%%%%%%%%%%%%%%%%%%%%%%%%%%%%%%%%%%%%%%%%%%%%%%%%%%%%%%%%%%%%%

In  McPhedran et al (2011)   the distributions along the critical line $\sigma=1/2$ of the zeros of
${\cal C}(0,1;s)$ were compared with those of the angular lattice sums ${\cal C}(1,4;s)$, ${\cal C}(1,8;s)$ and 
${\cal C}(1,12;s)$.  Data was presented  for $t$ in the range $[1,300]$ which suggested that all four functions had the same distribution of zeros on this line, and gave a proof that this was the case if 
${\cal C}(0,1;s)$ obeyed the Riemann hypothesis  (which implied that the three angular lattice sums obeyed it as well). As we have remarked, P. R. Taylor (1945) proved that ${\cal T}_-(s)$
obeyed the Riemann hypothesis, and we have seen the connection between ${\cal C}(0,1;s)$ 
and ${\cal T}_+(s)$ in equation (\ref{n25a}) above. It is therefore of interest to examine the
distributions of zeros along the critical line of  ${\cal T}_-(s)$ and ${\cal T}_+(s)$, and to compare these with the distribution of  ${\cal C}(0,1;s)$. We also compare these distributions of zeros
with that of $\zeta(2 s-1/2)$.

As a motivation for the subsequent analytic derivations, we first examine numerically the distributions of zeros of these four functions. Table 1 gives the distributions, with the data for
${\cal C}(0,1;s)$ being taken from Table 1 of  McPhedran et al (2011). The comparison of the second, third and fourth columns of Table 1 clearly suggests that in fact the distributions of zeros for ${\cal T}_-(s)$ and ${\cal T}_+(s)$ are the same as those for ${\cal C}(0,1;s)$,  $\zeta(2 s-1/2)$  and ${\cal C}(1,4;s)$, ${\cal C}(1,8;s)$ and 
${\cal C}(1,12;s)$. Table 2 extends the comparison of zero distributions for    these  four functions  to higher values of $t$.  To compile Table 2, a listing of the first ten thousand zeros of $L_{-4}(s)$  (Silva, 2007) has been used.

 \begin{table}
\caption{Numbers of zeros of ${\cal C} (0,1;1/2+\ri t)$, ${\cal T}_-(1/2+ \ri
t)$ , ${\cal T}_+(1/2+\ri t)$ and $\zeta(1/2+2 i t)$  in successive intervals of $t$.}
\begin{center}\begin{tabular}{|c|c|c|c|c|} \\ \hline
$t$ & $n({\cal C})$ & $n({{\cal T}_-})$ & $n({{\cal T}_+})$ &   $n( \zeta(2 s-1/2))$       \\ \hline
0-10 & 1 & 1 & 1 &1 \\
10-20 & 5 & 5& 5 & 5\\
20-30 & 7&7&7  &7\\
30-40 & 7&7&8 &8\\
40-50 & 10&9&8 &8 \\
50-60 & 8 &9&9 &9 \\
60-70 & 10 & 9 &10 &10\\
70-80  & 10 & 11 & 10 &10\\
80-90 & 11 &10 &11 &11\\
90-100 & 10 & 11 &10  &10\\
\hline 0-100 & 79 & 79 & 79 & 79\\ \hline
100-110 & 11&11&11 &11 \\
110-120 & 12 &11 &12  & 12\\
120-130 & 12 &12 &12 &12 \\
130-140 & 12 &12 &12 &12 \\
140-150 & 11 &12 &12 &12 \\
150-160 & 13 &13 &12 &12    \\
160-170 &  13 &12 &13 & 13     \\
170-180 & 14 &13 &13 & 13  \\
180-190 & 13 &13 &13 &13  \\
190-200 & 12 &13 &13 &13  \\
\hline 0-200 & 202 &201 &202 & 202 \\ \hline
200-210 & 14 &13 &13 & 13\\
210-220 & 13 &14 &14& 14 \\
220-230 & 14 &14 &13 &13\\
230-240 & 14 &14 &14 &14  \\
240-250 & 14 &13 &13  &13\\
250-260 & 14 &14  &15 & 14   \\
260-270 & 14 &14 &14  &14  \\
270-280 & 14 &15 &14  &15 \\
280-290 & 14 &14& 14 & 14 \\
290-300 & 14 &14  & 15 &15 \\
\hline 0-300 & 342 & 340 & 341 &341 \\ \hline
\end{tabular}
\end{center}
\label{table1}
\end{table}

 \begin{table}
\caption{Numbers of zeros of ${\cal C} (0,1;1/2+\ri t)$ and $\zeta(1/2+2 i t)$  in successive intervals of $t$.}
\begin{center}\begin{tabular}{|c|c|c|c|c|c|c|c|} \\ \hline
$t$ &$n({\cal C})$ & $n( \zeta(2 s-1/2))$ & $n({\cal T}_-)$ & $n({{\cal T}_+})$ &  
   $t$ &$n({\cal C})$ &   $n( \zeta(2 s-1/2))$       \\ \hline
0-100&79&79&79&79&1000-2000 &1958 &1957 \\
100-200&122&123&122&123& 2000-3000&2124 & 2124 \\
200-300&140&139&139& 139&3000-4000 & 2231 &2232 \\ 
300-400&150  &150 & 150& 150&4000-5000  &2313 & 2312 \\
400-500 & 157 &158 &158& 158 &5000-6000 &2376 &2377 \\
500-600 &166 &164 &165 &164& 6000-7000 & 2431 &2431 \\
600-700 & 168 &170  & 169&170&7000-8000 & 2474 &2475 \\
700-800 &176 &174 &175& 174 &0-8000  & 17424 & 17425 \\
800-900 &178 &178&178&178 & &  &\\
900-1000& 181&182 &182&182 & &  &\\
0-1000 &1517 &1517&1517&1517 & & & \\ \hline
\end{tabular}
\end{center}
\label{table2}
\end{table}

We recall from  McPhedran et al (2011) the formulae for the distributions of zeros of $\zeta (s)$  and ${\cal C}(0,1;s)$ (assuming the Riemann hypothesis to hold for both):
\begin{equation}
N_{\zeta}(\frac{1}{2},t)=\frac{ t}{2\pi} \log (t)-\frac{
t}{2\pi}(1+\log(2 \pi)) +\ldots,
\label{dz3}
\end{equation}
and
\begin{equation}
N_{{\cal C}0,1}(\frac{1}{2},t)=\frac{ t}{\pi} \log (t)-\frac{
t}{\pi}(1+\log( \pi)) +\ldots .
\label{dz5}
\end{equation}
It is an interesting fact that if $t$ is replaced by $2 t$ in (\ref{dz3}) we obtain
\begin{equation}
N_{\zeta}(\frac{1}{2},2 t)=\frac{ t}{\pi} \log (t)-\frac{
t}{\pi}(1+\log( \pi)) +\ldots,
\label{dz3a}
\end{equation}
in agreement with (\ref{dz5}) (which combined the distributions of zeros for $\zeta(s)$ 
and $L_{-4}(s)$).  Numerical exemplification of (\ref{dz5}) and (\ref{dz3a}) is to be found in Tables 1 and 2, in the latter case extending over the first 17 thousand zeros.

\section{Phase Distributions and Order of Zeros}

In order to begin our analytic investigations of  the zeros  of ${\cal T}_+(s)$ and ${\cal T}_-(s)$, we follow the method of  McPhedran et al (2011), and consider the properties of the
quotient function
\begin{equation}
 {\cal V}(s)=\frac{{\cal T}_+(s)}{{\cal T}_-(s)}.
\label{del51}
\end{equation}
\begin{theorem}
The analytic function ${\cal V}(s)$ is odd under $s\rightarrow 1-s$, it has a zero at $s=1/2$, is finite at $s=1$, and  has poles at  $s\approx 3.9125$,  $s\approx -2.9125$. It is pure imaginary on the critical line. It can be evaluated by direct summation in $\sigma>1$.  ${\cal V}(s)$ has an argument which lies in the fourth quadrant for $\sigma>2$ and $t>3$,
while in  $\sigma< -1$ and $t>3$ it   lies in the third quadrant.
\label{th51t}
\end{theorem}
\begin{proof}
Since ${\cal T}_+(s)$ is even under $s\rightarrow 1-s$ while  ${\cal T}_-(s)$ is  odd, it follows that
 ${\cal V}(s)$ is odd under $s\rightarrow 1-s$. This means that its real part is odd  and its imaginary part even under $\sigma+i t\rightarrow (1-\sigma) +i t$, a reflection in the critical line, so the first quadrant maps onto the second quadrant and the third onto the fourth under the reflection. The modulus of ${\cal T}_-(s)$ is even under the reflection. From these remarks, it follows that ${\cal V}(1/2+i t)$  is pure imaginary.
 
 At particular points: ${\cal V}(1)=-1$,  ${\cal V}(0)=1$ and  ${\cal V}(1/2+\delta)=-1.95381 \delta +O(\delta^2)$.  The function ${\cal V}(s)$ has only one derivative zero on the critical line, near $t=2.94334$ (see Fig. \ref{fig03t}; the proof is given as Corollary 4.4).
 
 \begin{figure}[h]
\includegraphics[width=3.0in]{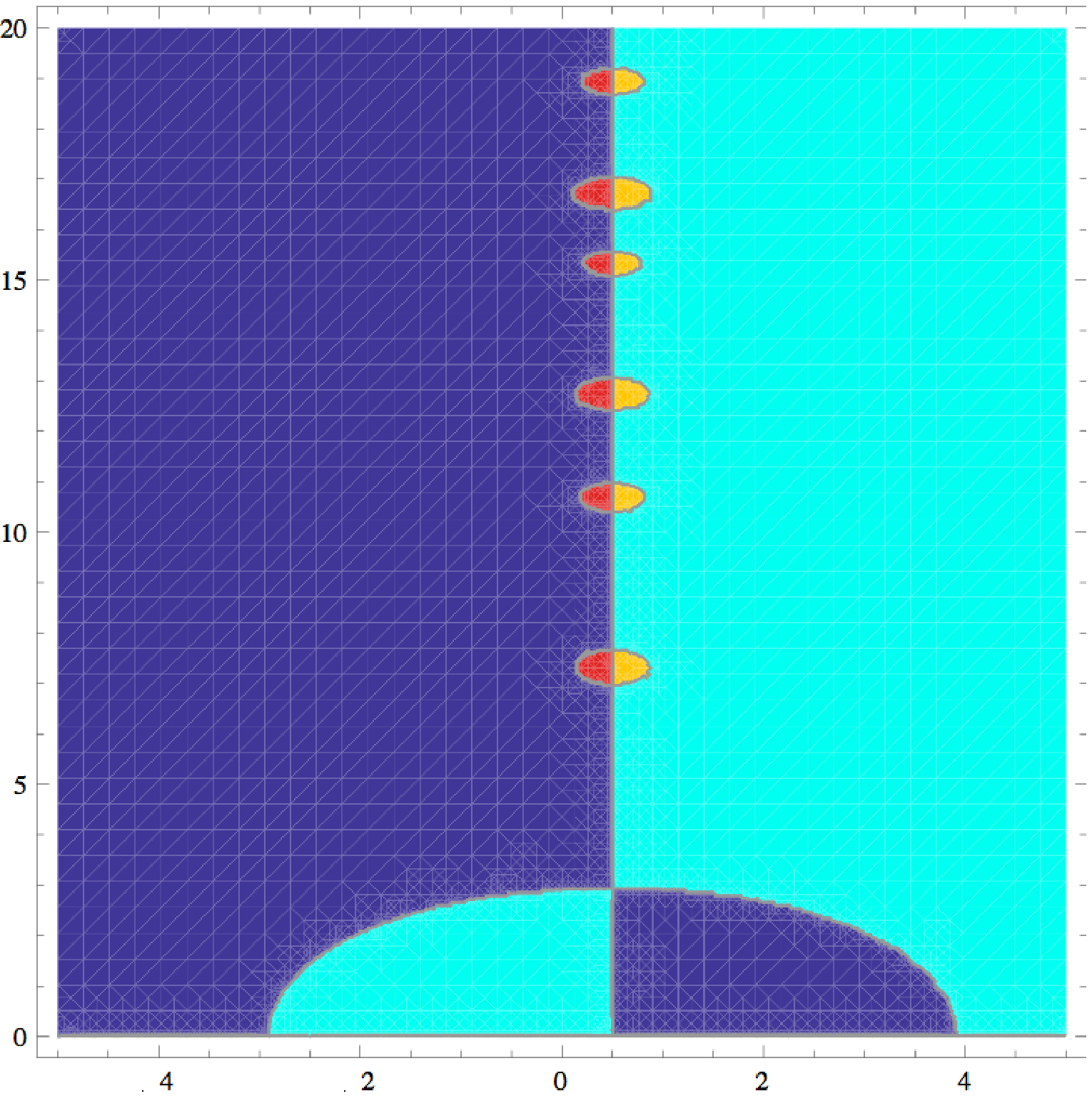}~~\includegraphics[width=3.0in]{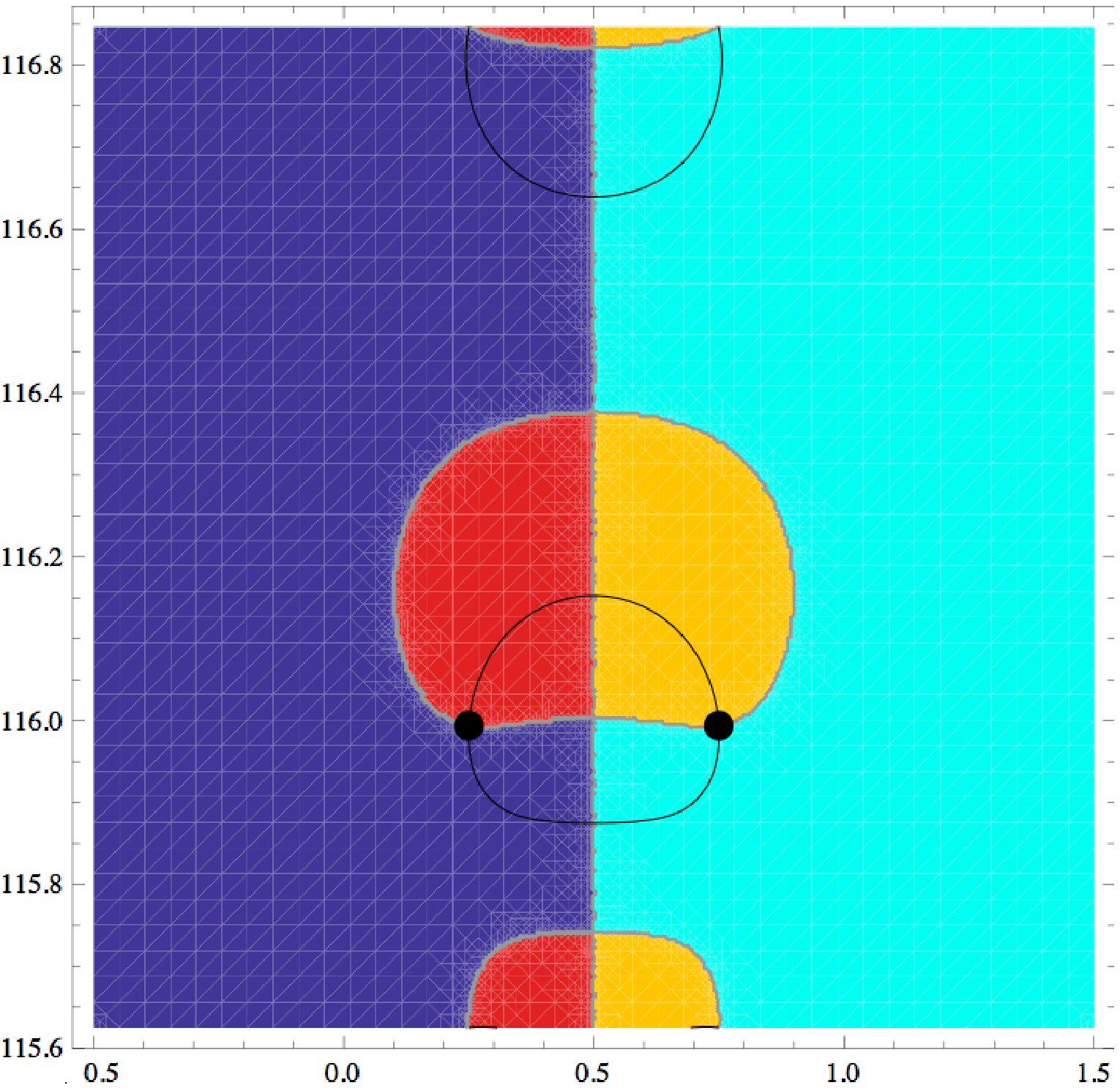}
\caption{\label{fig03t}  Plots of the phase of ${\cal V}(s)$ in the complex plane. The phase is denoted by colour according to which
quadrant it lies in:  yellow-first quadrant,  red- second, purple- third, and light blue- fourth.
The plot at right is in the vicinity of the 98th zero of $\zeta(2 s)$, denoted by a dot (together with its reflection in the critical line). The black contours  through the dots show lines along which $|{\cal V}(s)|=1$. }\end{figure}

 We write
 \begin{equation}
 {\cal V}(s)=\frac{1+{\cal U}(s)}{1-{\cal U}(s)},
 \label{tdef}
 \end{equation}
 where
 \begin{equation}
 {\cal U}(s)=\frac{\xi_1(2 s-1) }{\xi_1(2 s)}=\frac{\Gamma(s-1/2)\sqrt{\pi}\zeta (2 s-1)}{\Gamma(s) \zeta (2 s)}.
 \label{udef2}
 \end{equation}
 We verify numerically the proposition that in  $\sigma>1$ ${\cal V}(s)$ has an argument which lies in the fourth quadrant for $t$ above 2.94334.  When we reach the region $t>>1$, we can then
 use Stirling's expansion for the $\Gamma$ function and expand the $\zeta$ function by direct summation in $\sigma>1$:
 \begin{equation}
  {\cal U}(s)=
  \sqrt{\frac{\pi}{s}} (1+\frac{3}{8 s}+\ldots) (1+\frac{1}{4^s}+\frac{2}{9^s}+\ldots).
 \label{udef3}
 \end{equation}
 If $\sigma>2$ (say), the second series in (\ref{udef3}) converges rapidly in exponential fashion, and the first series and its prefactor dominate, ensuring that ${\cal U}(s)$ and consequently
 ${\cal V}(s)$ have arguments in the fourth quadrant.
 For $\sigma<-1$, the symmetry properties of ${\cal V}(s)$ under reflection in the critical line ensure its argument lies in the third quadrant.
\end{proof}

\begin{theorem}
The analytic function ${\cal U} (s)$  has a phase which is even under $s\rightarrow 1-\bar{s}$ (with the bar denoting complex conjugation) and a modulus which reciprocates. It has a zero at $s=0$,  a pole at $s=1$, and (apart from the influence of the pole) decreases monotonically as $s=\sigma$ goes along the real axis from minus infinity to infinity. Apart from the possible exception of points that are poles of ${\cal U}(s)$,
${\cal V}'(s)=0$  implies ${\cal U}'(s)=0$.  ${\cal U}(s)$ can be evaluated by direct summation in $\sigma>1$.  ${\cal U}(s)$ has an argument which lies in the fourth quadrant for $\sigma>1$ or  $\sigma< 0$  and $t>0.5249$.
\label{th51u}
\end{theorem}
\begin{proof}
The definition (\ref{udef2}), together with the functional equation for the Riemann zeta function
(Titchmarsh \& Heath-Brown, 1987),
ensure that ${\cal U}(s)$ obeys the functional equation
\begin{equation}
{\cal U}(s) {\cal U}(1-s)=1.
\label{uthm1}
\end{equation}
This, together with its analytic nature, ensures that its phase is symmetric under the reflection
$s\rightarrow 1-\bar{s}$, while
its modulus reciprocates (i.e. $\log[|{\cal U}(s)|]$ is antisymetric under reflection).

Its behaviour near $s=1$ is governed by the singularity of the Riemann zeta function, and we find
\begin{equation}
{\cal U}(1+\delta)=\frac{3}{\pi \delta}+O(\delta^0).
\label{uthm2}
\end{equation}
Using (\ref{uthm1}) we then have
\begin{equation}
{\cal U}(\delta)=-\frac{\pi \delta}{3}+O(\delta^2).
\label{uthm3}
\end{equation}

The inverse of (\ref{tdef}) is
\begin{equation}
{\cal U}(s)=\frac{{\cal V}(s)-1}{{\cal V}(s)+1}.
\label{uthm4}
\end{equation}
The fixed points (Knopp, 1952)  of both ${\cal U}$ and ${\cal V}$ are $\pm i$, so that the two functions are also related in the form
\begin{equation}
F_1(s)=\left( \frac{{\cal V}(s)-i}{{\cal V}(s)+i}\right)=i\left( \frac{{\cal U}(s)-i}{{\cal U}(s)+i}\right).
\label{uthm5}
\end{equation}
Note that from (\ref{uthm4}), lines of phase zero or $\pi$ of ${\cal V}$ map onto lines of phase
zero or $\pi$ of ${\cal U}$, and vice versa. The exact correspondence is determined by the modulus
of the first-mentioned function. Also, the lines of $|{\cal U}|=1$ correspond to the lines 
 $\arg[{\cal V}(s)]=\pm \pi/2$
where ${\cal V}(s)$ is pure imaginary, and vice versa. The relationship between ${\cal V}(s)$
and ${\cal U}(s)$ is, from equations (\ref{uthm4},\ref{uthm5}), a rotation through $90^\circ $
in the complex plane (Cohn, 1967).

The derivative of equation (\ref{uthm4}) is
\begin{equation}
{\cal U}'(s)=\frac{2{\cal V}'(s)}{[{\cal V}(s)+1]^2}.
\label{uthm4a}
\end{equation}
This shows that, if ${\cal V}(s)\neq -1$, ${\cal V}'(s)=0$ implies ${\cal U}'(s)=0.$

For $t$ near zero, the phase of ${\cal U}(s)$ lies in the third quadrant in a  region which at $t=0$ extends between the pole at $s=1$ and the zero at $s=0$. It terminates on the critical line at the value $t=0.5249$, determined numerically.

We have given the expansion of ${\cal U}(\sigma+i t)$ for $t>>1$ and $\sigma>1$ in (\ref{udef3}).
From this, it follows that as $\sigma$ becomes large compared with $t$, $|{\cal U}(\sigma+i t)|$ tends to zero as $\sqrt{\pi/\sigma}$, and for $\sigma$ tending to $-\infty$ it diverges as $\sqrt{\sigma/\pi}$. In either case, $|{\cal V}(\sigma+i t)|$ tends to unity.

\end{proof}

We will show below that the function $F_1(s)$ provides a bridge from the properties of ${\cal V}(s)$ to those of ${\cal U}(s)$. Hence, it is necessary to establish its key properties.
\begin{theorem}
The function $F_1(s)$ satisfies the functional equation
\begin{equation}
F_1(1-s)=\frac{1}{F_1(s)}.
\label{f1-1}
\end{equation}
It is real on the critical line, has unit modulus on the real axis, has the special values
$F_1(0)=-i$, $F_1(1/2)=-1$, and in $t>>\sigma>>1$ has the asymptotic expansion
\begin{equation}
F_1(\sigma+i t)=-i +(1-i)\sqrt{\frac{2 \pi}{t}}+\frac{2 \pi}{t}+(1+i)(\pi+\frac{\sigma}{2})\sqrt{\frac{2 \pi}{t^3}}+\ldots .
\label{f1-2}
\end{equation}
\label{thmf1}
\end{theorem}
\begin{proof}
The functional equation for $F_1(s)$ follows readily from that for ${\cal U}(s)$. It is real on the critical line 
since ${\cal V}(s)$ is pure imaginary there. Again, since ${\cal V}(\sigma)$ is real, the modulus of
${\cal F}_1(\sigma)$ is unity. Its special values at zero and $1/2$ are in keeping with this, and with
the functional equation in the latter case. The asymptotic expansion for $F_1(\sigma+i t)$ follows readily from equation (\ref{udef3}), neglecting exponentially small terms.

Note that from (\ref{f1-2}), the phase of $F_1(s)$ will lie in the fourth quadrant in the asymptotic region to the right of the critical line, and, using the functional equation (\ref{f1-1}), in the first quadrant in the asymptotic region to the left of it. Note also that on the critical line
\begin{equation}
F_1'(\frac{1}{2}+ i t)=\frac{2 [\Im ({\cal V}(\frac{1}{2}+i t))]'}{(1+\Im({\cal V}(\frac{1}{2}+i t)))^2}.
\label{f1-3}
\end{equation}
which is real.
\end{proof}

\begin{corollary} On the critical line, the functions ${\cal U}(s)$, ${\cal V}(s)$ and $F_1(s)$ are  related simply to the symmetrized zeta function in the following way:
\begin{eqnarray}
\label{ec1}  {\cal U}(\frac{1}{2}+ i t)=\exp[-2 i \arg[\xi_1(1+2 i t)]] , &
{\cal V}(\frac{1}{2}+ i t)=i\tan[\arg[\xi_1(1+2 i t)]+\frac{\pi}{2}], \nonumber \\ 
 F_1(\frac{1}{2}+ i t)=\tan[\arg(\xi_1(1+2 i t))+\frac{\pi}{4}]. & 
 \label{ecor1}
\end{eqnarray}
\label{extracor}
\end{corollary}
\begin{proof}
We recall the definition (\ref{udef2}) of ${\cal U}(s)$:
\begin{equation}
{\cal U}(s)=\frac{\xi_1(2 s-1) }{\xi_1(2 s)}.
\label{ecor3}
\end{equation}
Putting $s=1/2+i t$ and using the symmetry of $\xi(s)$ under $s\rightarrow 1-s$, equation (\ref{ecor3}) yields
\begin{equation}
{\cal U}(\frac{1}{2}+ i t)=\frac{\xi_1(2 i t) }{\xi_1(1 + 2 i t))}=\frac{\overline{\xi_1(1+2 i t)} }{\xi_1(1+2 i t)}.
\label{ecor4}
\end{equation}
This is equivalent to the first of equations (\ref{ecor1}).
The second and third of equations (\ref{ecor1}) follow readily from the first.
\end{proof}

It follows from equations (\ref{ecor1}) that, if $t_1$ and $t_2$ are any two values of $t$ such that
$\arg[\xi_1(1+2 i t_2)]=\arg[\xi_1(1+2 i t_1)]+n \pi$, for any integer $n$, then ${\cal U}(\frac{1}{2}+ i t_2)={\cal U}(\frac{1}{2}+ i t_1)$, ${\cal V}(\frac{1}{2}+ i t_2)={\cal V}(\frac{1}{2}+ i t_1)$
and $F_1(\frac{1}{2}+ i t_2)=F_1(\frac{1}{2}+ i t_1)$.

The properties of ${\cal V}(s)$ , ${\cal U}(s)$   and $F_1(s)$ proved in Theorems \ref{th51t}, \ref{th51u} and \ref{thmf1}
are exemplified in Figs. \ref{fig98} and \ref{fig1444}. The former shows the behaviour of the distributions of phase for the three functions around the 98th zero of $\zeta(2 s-1/2)$. It will be noted that the three plots show the same lines bounding phase regions, but with
differing interpretations for each function. For example, the closed contour formed by the segments of constant phase $0$ and $\pi$ in the first case (${\cal V}(s)$), are preserved in the second case (${\cal U}(s)$), while the contour of constant amplitude $|{\cal V}(s)|=1$ has become a contour composed of
segments of constant phase ($\arg[{\cal U}(s)]=\pm \pi/2$). The region containing the pole and zero of $F_1(s)$ is more compact than the corresponding regions for the other two functions.

\begin{figure}[h]
\includegraphics[width=3.0in]{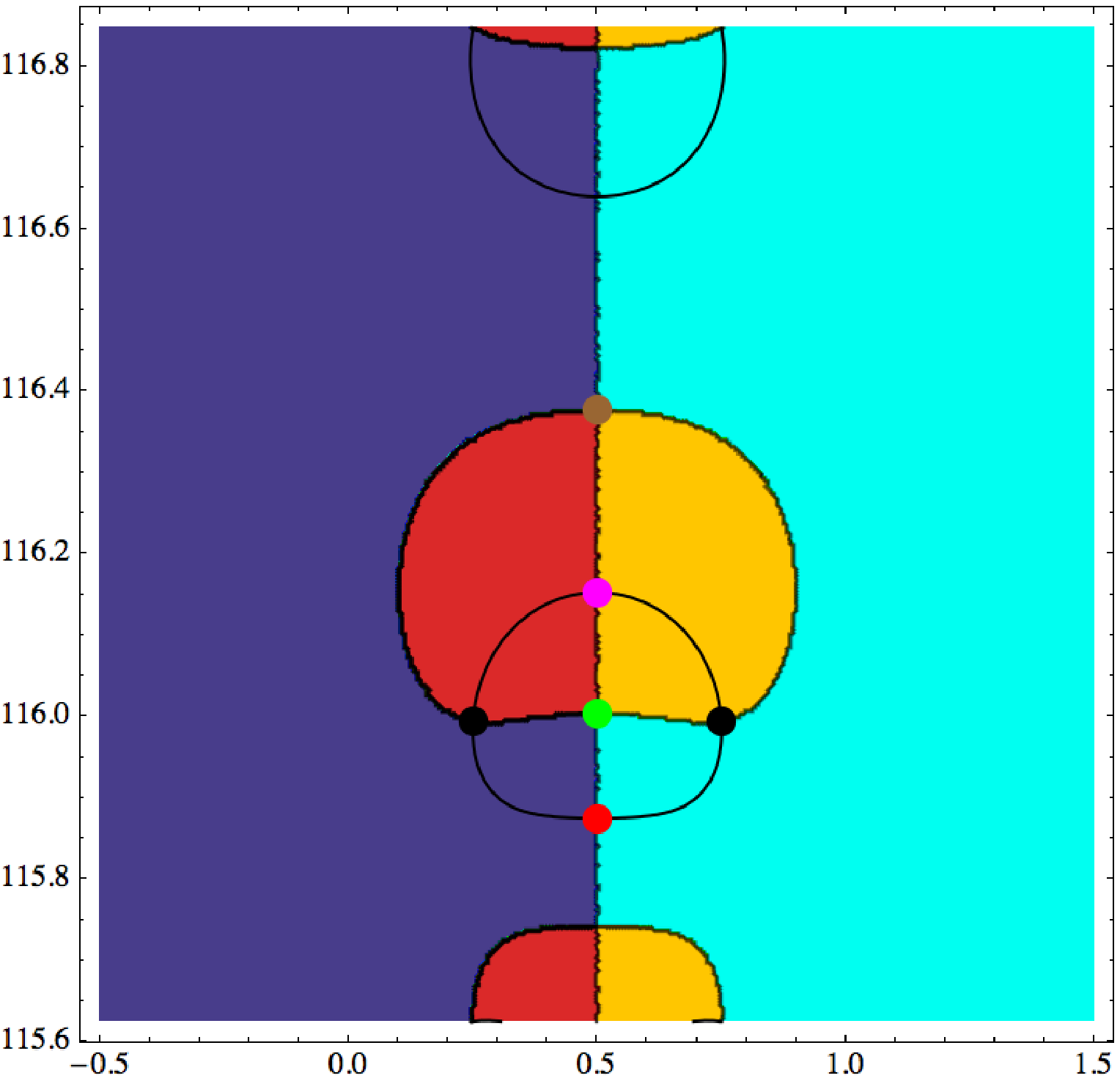}~~\includegraphics[width=3.0in]{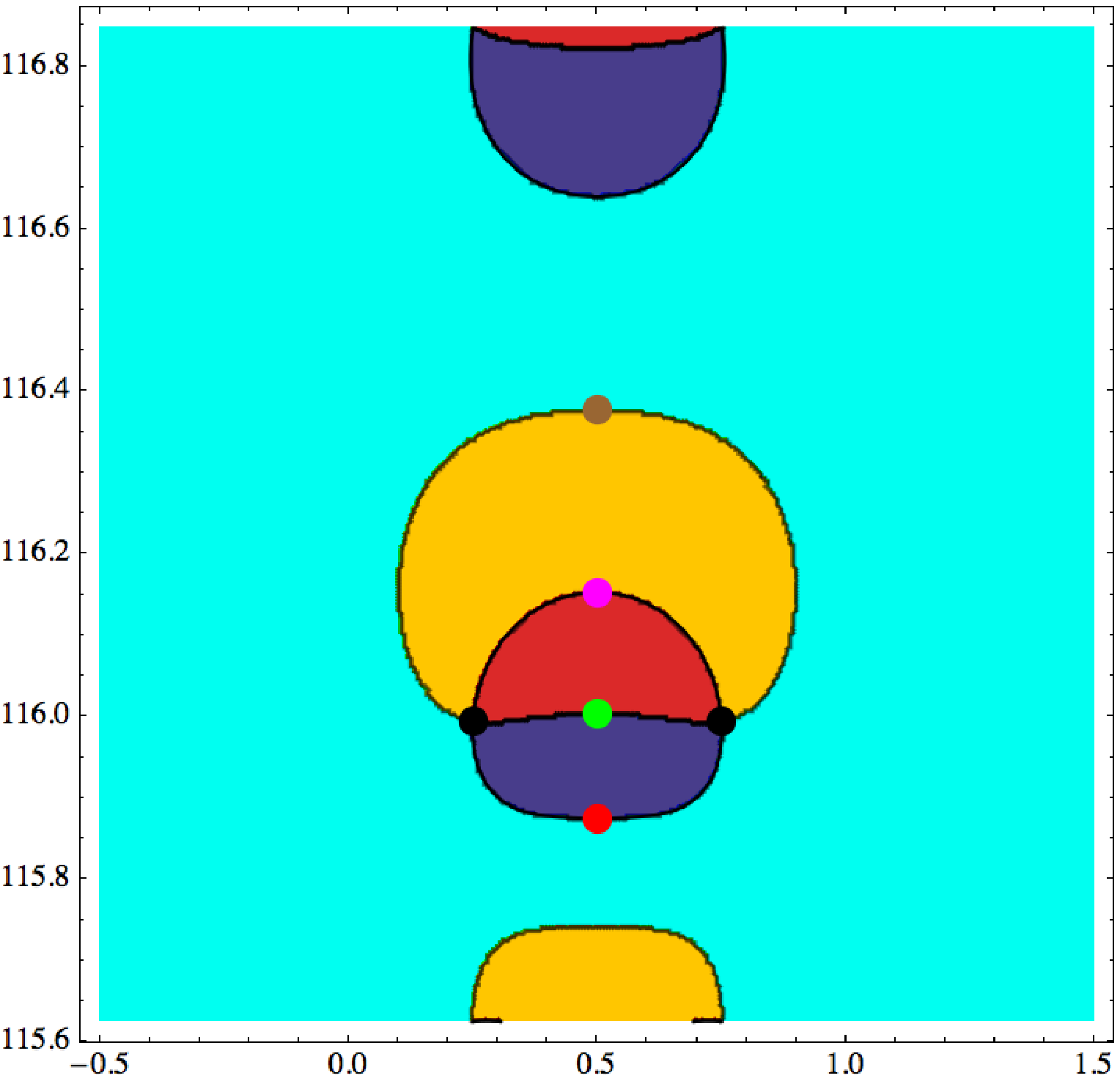}\\
\includegraphics[width=3.0in]{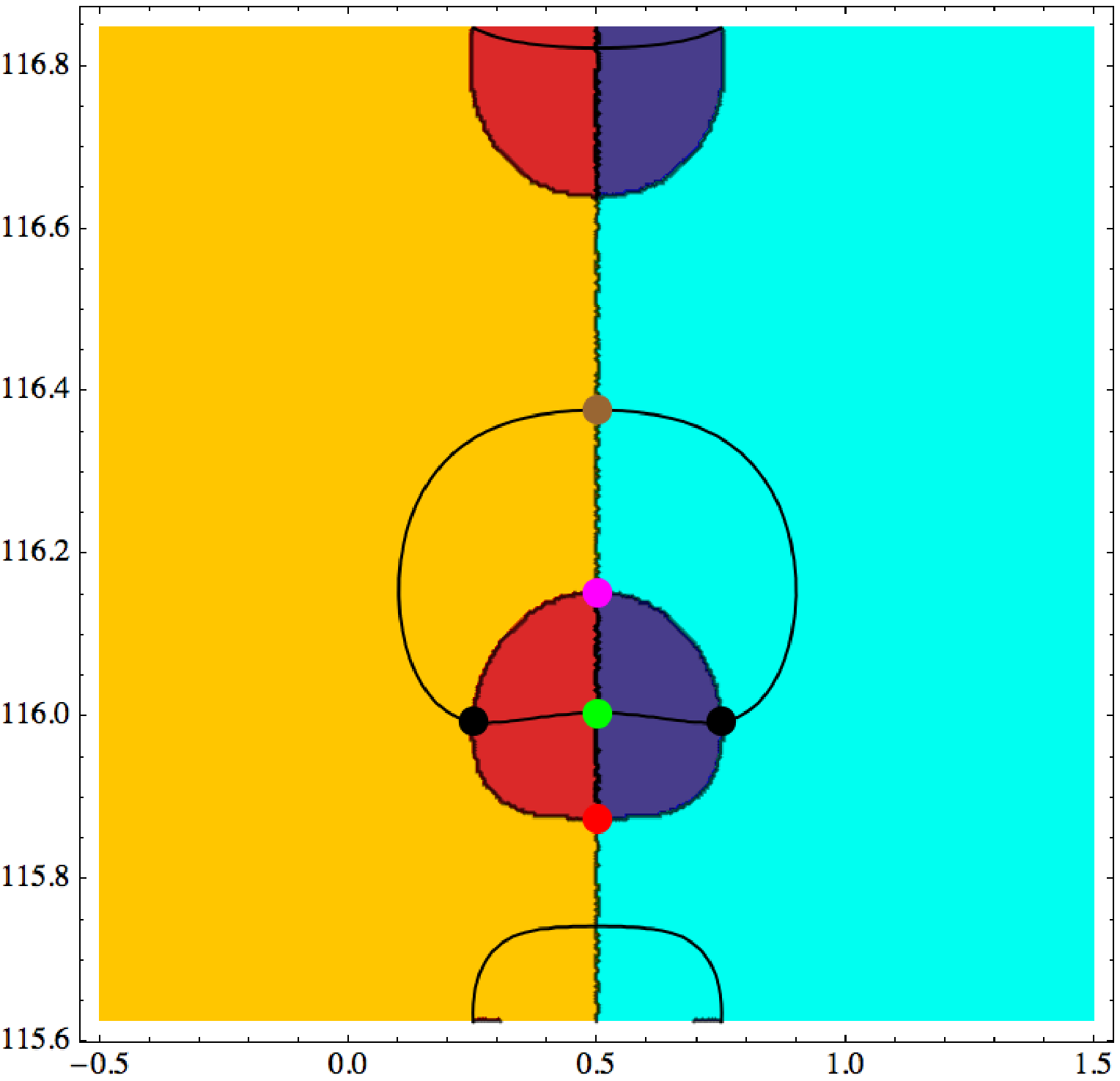}
\caption{\label{fig98}  Plots of the phase of ${\cal V}(s)$ (top left), ${\cal U}(s)$ (top right)
and $F_1(s)$  in the complex plane, in the vicinity of the 98th zero of $\zeta(2 s)$, denoted by a dot (together with its reflection in the critical line).  The phase is denoted by colour according to which
quadrant it lies in:  yellow-first quadrant,  red- second, purple- third, and light blue- fourth.
 The black contours  through the dots show lines along which $|{\cal V}(s)|=1$ (top two  plots)
 and $|F_1(s)|=1$ (lower plot). The coloured dots are: zero (green) and pole (brown) of
 ${\cal V}(s)$; pole (red) and zero (magenta) of $F_1(s)$.  }\end{figure}

Figure \ref{fig1444} shows the most extreme variation among the 1517 phase plots constructed in the range of values of $t$ up to 1000. The phase plots around the 1444th and 1445th zeros of $\zeta(2s -1/2)$ have become linked, more strongly in the cases of ${\cal V}(s)$ and ${\cal U}(s)$
than for $F_1(s)$. The two former functions have phase distributions with inner and outer bounding lines, with the zeros and poles of ${\cal U}(s)$ occurring on both boundaries.

\begin{figure}[h]
\includegraphics[width=3.0in]{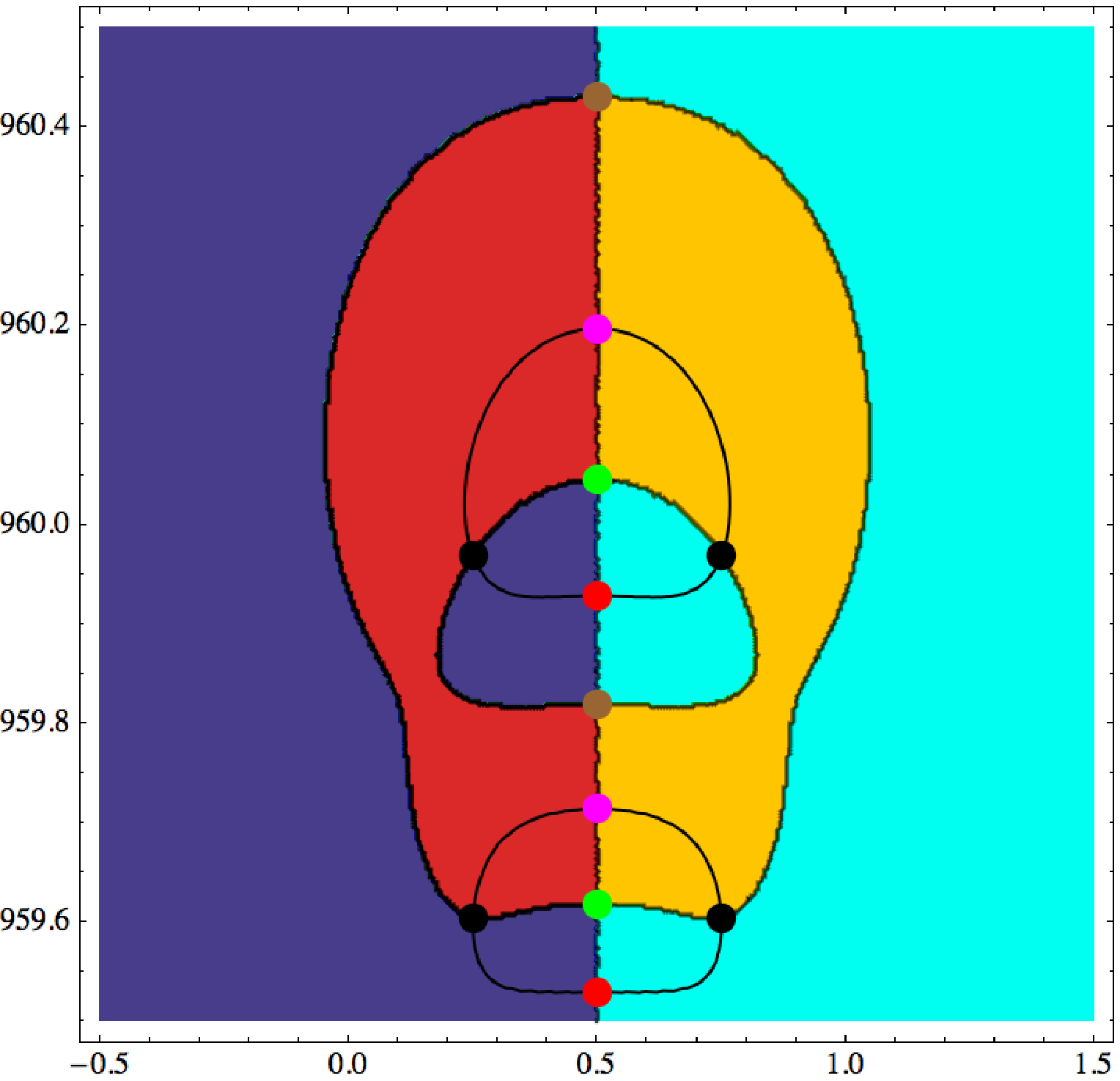}~~\includegraphics[width=3.0in]{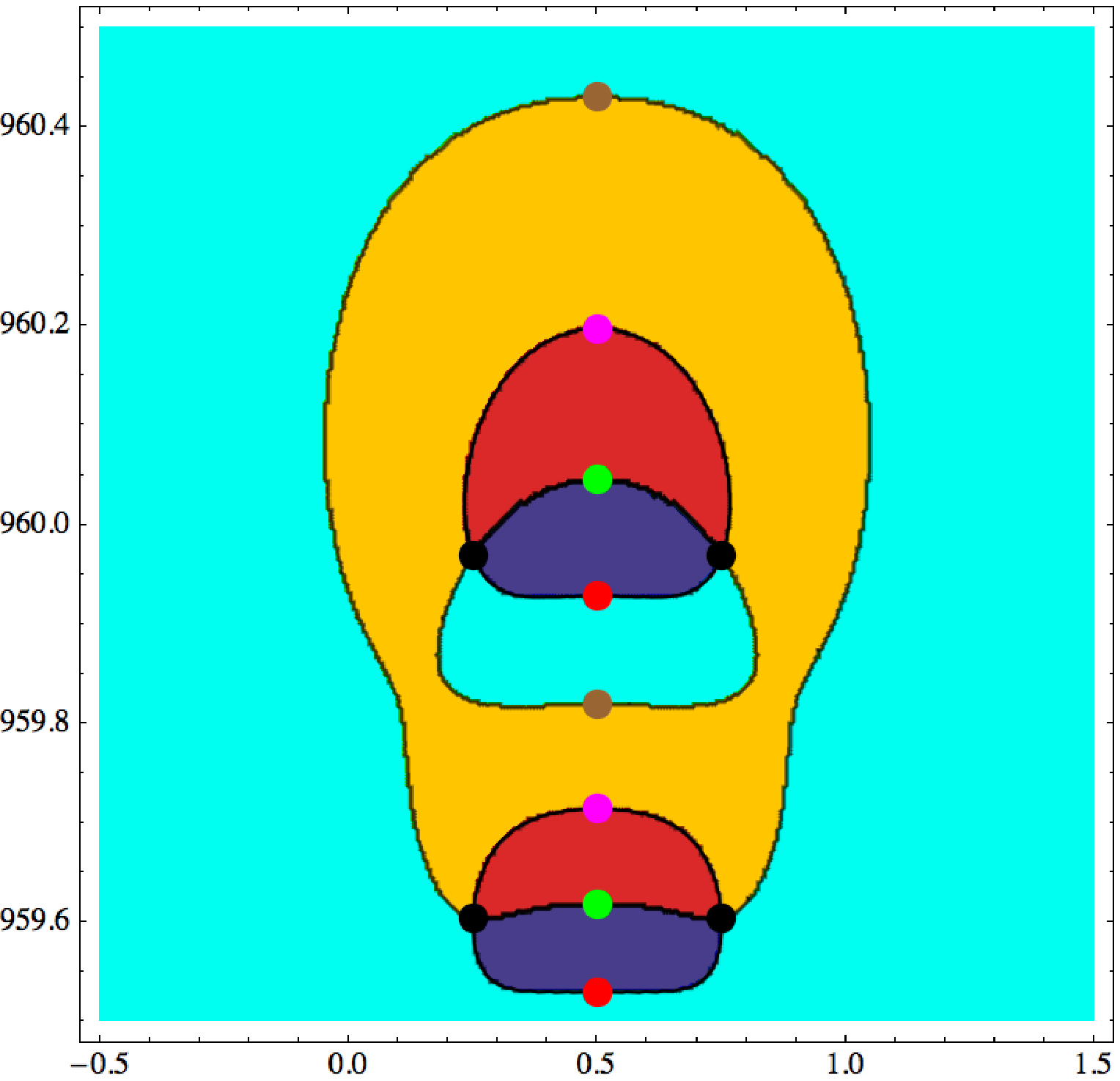}\\
\includegraphics[width=3.0in]{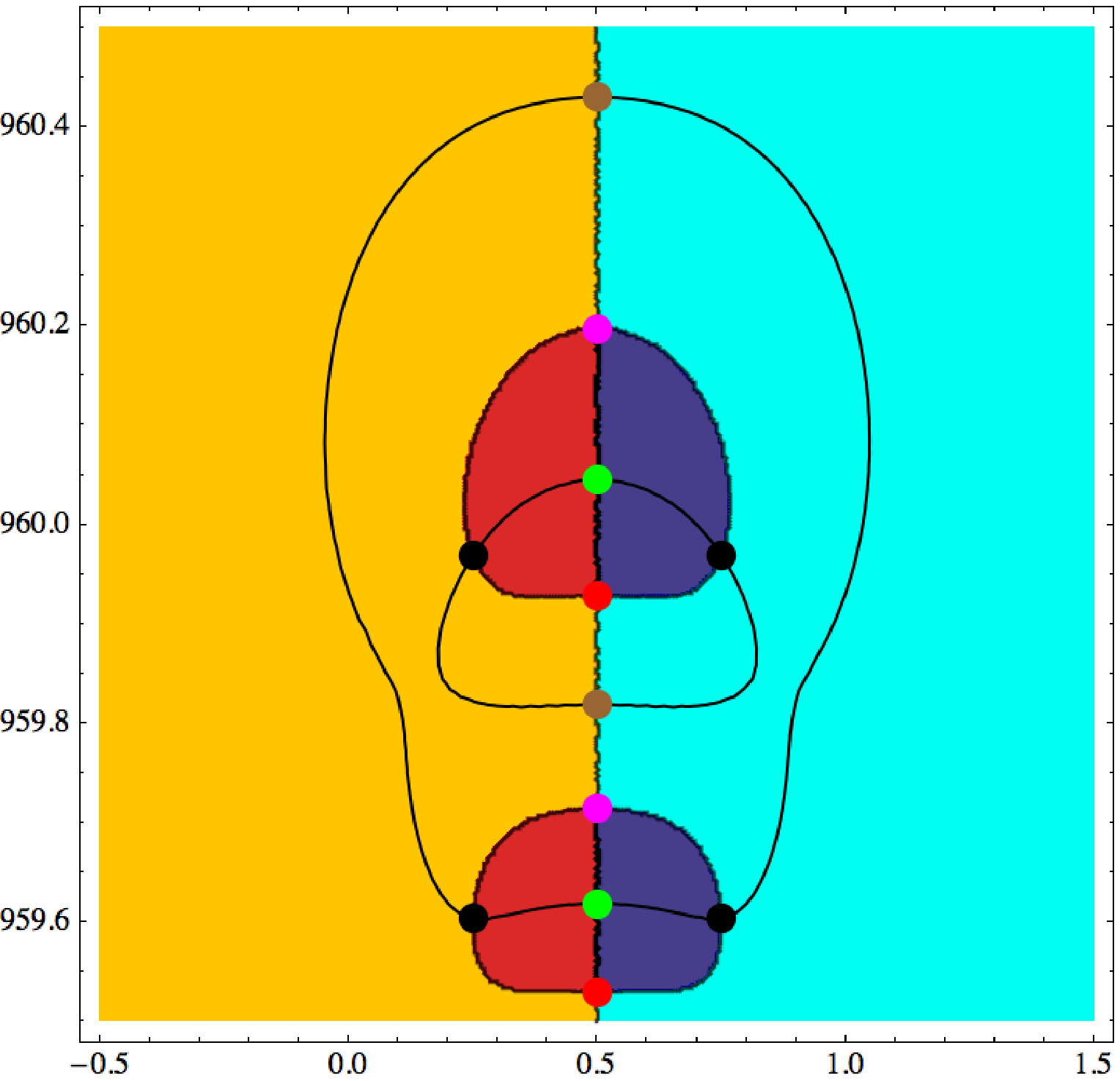}
\caption{\label{fig1444}  Plots of the phase of ${\cal V}(s)$ (top left), ${\cal U}(s)$ (top right)
and $F_1(s)$  in the complex plane, in the vicinity of the 1444th and 1445th zeros of $\zeta(2 s)$, denoted by  dots (together with their reflections in the critical line).  . The phase is denoted by colour according to which
quadrant it lies in:  yellow-first quadrant,  red- second, purple- third, and light blue- fourth.
 The black contours  through the dots show lines along which $|{\cal V}(s)|=1$ (top two  plots)
 and $|F_1(s)|=1$ (lower plot). The coloured dots are: zero (green) and pole (brown) of
 ${\cal V}(s)$; pole (red) and zero (magenta) of $F_1(s)$.  }\end{figure}

\begin{figure}[h]
\includegraphics[width=3.0in]{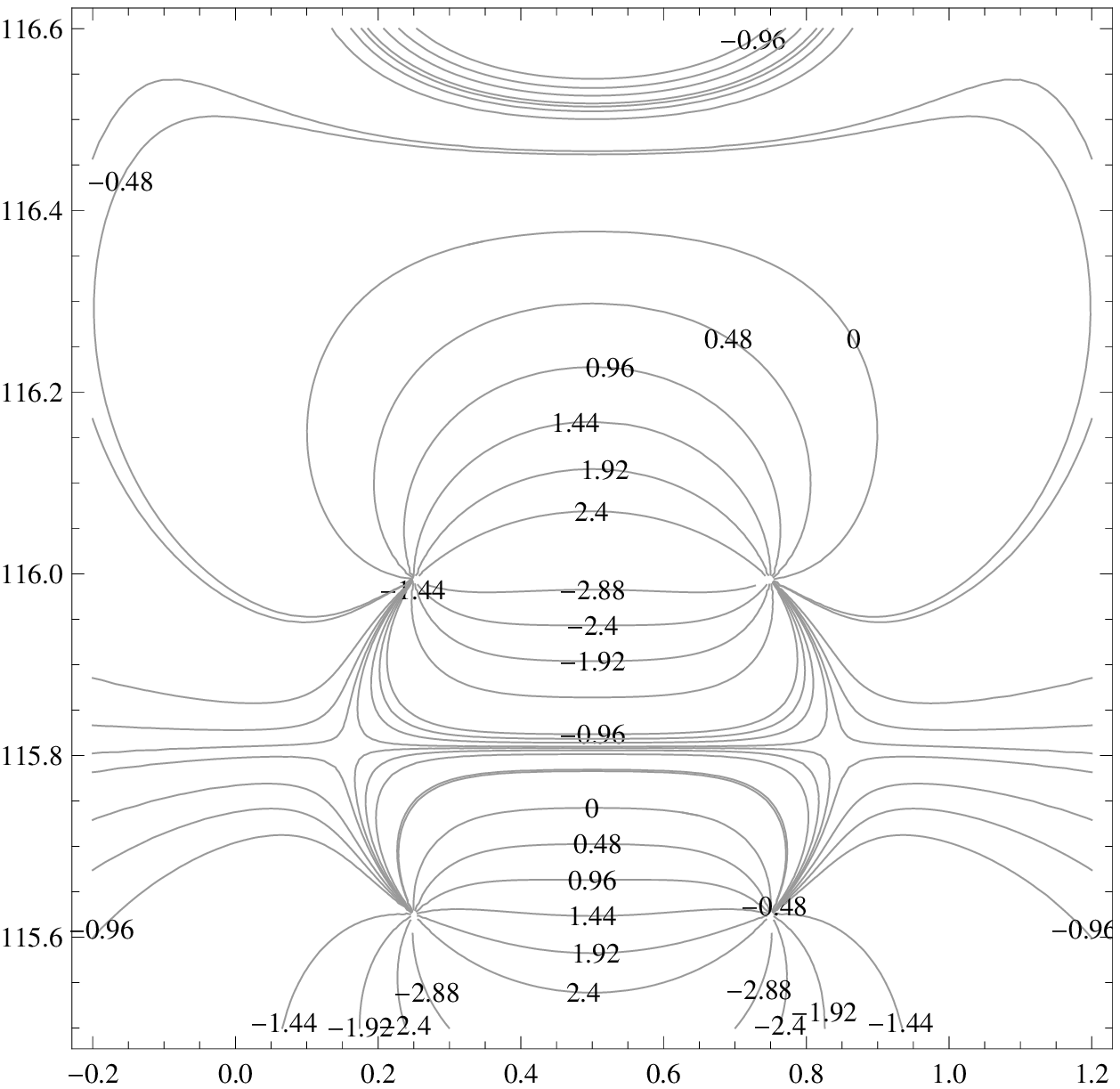}~~\includegraphics[width=3.0in]{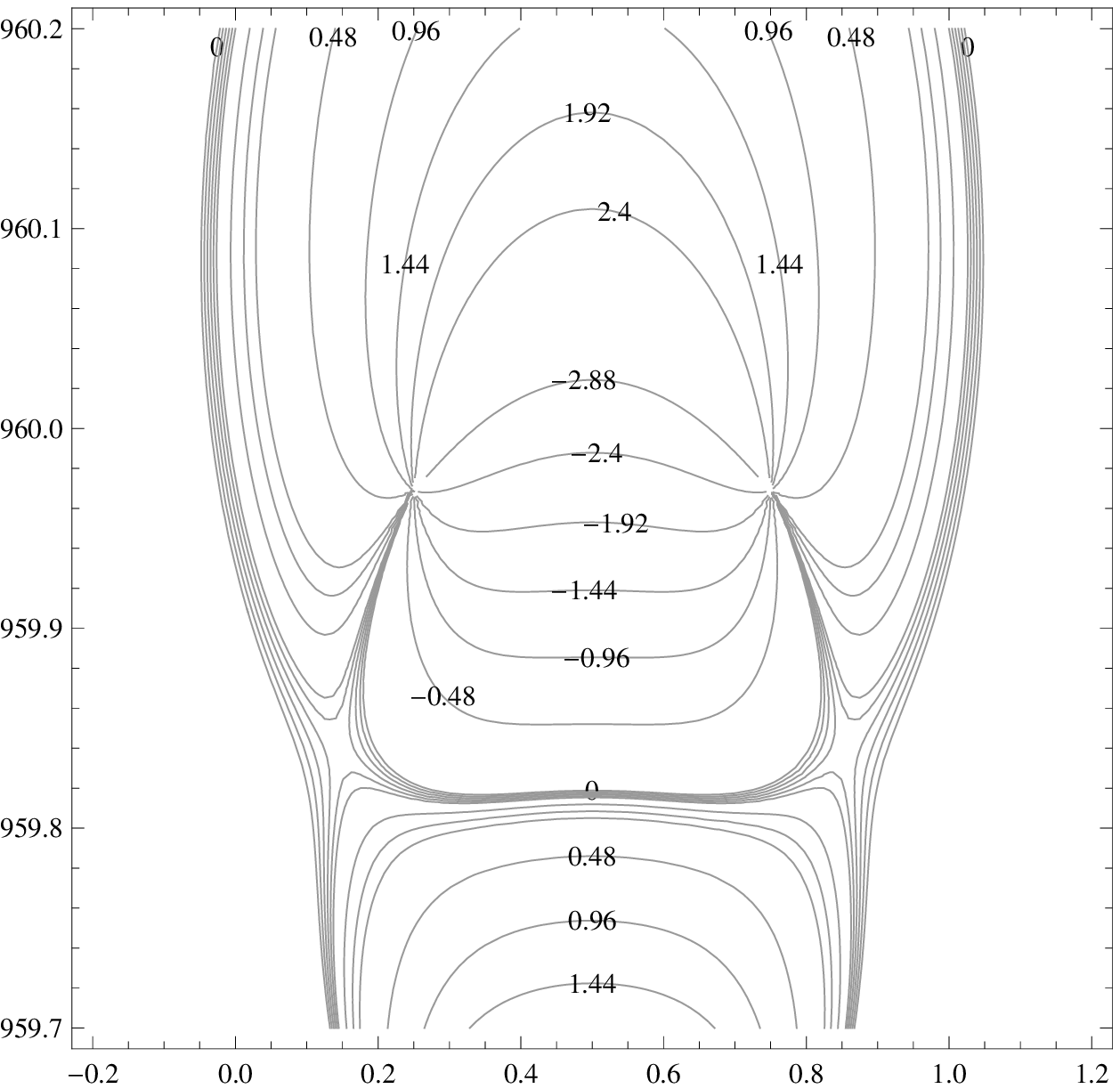}
\caption{\label{fig-derivz}  Plots showing lines of constant phase of ${\cal U}(s)$ in the vicinity of zeros of ${\cal U}'(s)$, corresponding to (left)  Fig. \ref{fig98} and (right) Fig. \ref{fig1444}.  
In the former case, the derivative zero occurs near 0.17+115.80 i, with $\arg {\cal U}(s)$ near -0.79, while in the latter case, it occurs near 0.15+959.84 i, with $\arg {\cal U}(s)$ near  0.036. }\end{figure}

We give in Fig. \ref{fig-derivz}  diagrams of contours of constant phase of ${\cal U}(s)$, showing the regions around points where ${\cal U}'(s)=0$. These points occur where lines of constant phase running from a pole of ${\cal U}(s)$ to a zero separate from lines of constant phase associated with the pole-zero pair lying above or below. In the typical case at left, the point where ${\cal U}'(s)=0$ corresponds to a phase in the fourth quadrant, while in the rare case at right the phase value lies in the first quadrant.

Our study of the zeros and poles of the three functions  ${\cal V}(s)$, ${\cal U}(s)$ 
and $F_1(s)$, for $t$ ranging up to 1000 suggest the following properties merit invesitgation:
\begin{itemize}
\item $\Im({\cal V} (1/2+i t)$ is a monotonic increasing function of $t$ above $t=2.94334$, with its increasing intervals interrupted only by sign reversals at its poles, which are of first order;
\item $|{\cal U}(\sigma+i t)|$ above $t=2.94334$ is always larger than unity to the left of the critical line, and smaller than unity to the right of the critical line;
\item along the critical line, we find groups of zeros and poles always occurring in the following order: pole of $F_1$, zero of ${\cal V}$, zero of  $F_1$, pole of ${\cal V}$;
\item if an upper limit on $t$ is chosen to end at a point where such a group is wholly included,
the numbers of zeros and poles of $F_1$ and ${\cal V}$ are exactly the same, and furthermore
agree with the number of zeros of $\zeta (2 s-1/2)$ if the upper limit also includes the associated pole and zero of ${\cal U}(s)$;
\item every pole (zero) of ${\cal U}(s)$ lies at the intersection  of a contour $|{\cal V}(s)|=1$
with a contour $\arg {\cal V}(s)=\pi$ ($\arg {\cal V}(s)=0$);
\item every pole (zero) of ${\cal U}(s)$ lies at the intersection  of a contour $|F_1(s)|=1$
with a contour $\arg F_1(s)=\pi/2$ ($\arg F_1(s)=-\pi/2$).
\end{itemize}
Proofs of all of these properties are put forward in Section 4 and the Appendices.

The alignment between values of $t$ for which ${\cal U}(s)$ has zeros and poles, and between corresponding values for ${\cal T}_+(s)$ and  ${\cal T}_-(s)$ is not subject to strict rules. The zero of
${\cal T}_-(s)$ lies below that of $\zeta (2 s-1/2)$ four times in the first 1517. The zero of $\zeta (2 s-1/2)$ lies below that of ${\cal T}_+(s)$ 235 times in  the first 1517 occasions.

Since the number of zeros of $\zeta(2 s-1/2)$, ${\cal T}_-(s)$ and ${\cal T}_+(s)$ up to $t=1000$
is the same, then the mean gap between zeros  for all three functions is the same (0.655).
The square root of the variance, expressed as a fraction of the mean gap is smallest for  ${\cal T}_-(s)$ (0.304), slightly larger for  ${\cal T}_+(s)$ (0.348), and more than twice as large for $\zeta(2 s-1/2)$ (0.655). This no doubt is a consequence of the trend for $\zeta (2 s -1/2)$ to vary more regularly with $t$ as $\sigma$ moves away from 1/2.

It is well known that the gaps between the zeros of the Riemann zeta function are distributed according to the Wigner surmise, following the pattern expected for the Gaussian unitary ensemble
(Dietz \& Zyczkowski 1991; Bogomolny \& Leboeuf 1994). Since the standard deviations of the gaps between zeros of ${\cal T}_-(s)$  and ${\cal T}_+(s)$ are distinctly smaller than that for
$\zeta(2 s-1/2)$, even though the mean gaps are the same, it is evident that the former two functions do not have gap distributions between zeros obeying the Wigner surmise.

A simple example having some of the properties of  ${\cal T}_+(s)$ has been commented on by Suzuki (2009). The analogue of  ${\cal T}_+(s)$  is $2 \cos (i(s-1/2))$; the natural analogue of ${\cal T}_-(s)$ is  $2 \sin (i(s-1/2))$, while that of ${\cal V}(s)$ is $\cot (i(s-1/2))$. This is real on the critical line rather than imaginary, but has first order zeros strictly alternating with first oder poles. The poles and zeros are however uniformly distributed along the critical line.

\section{Bounds on the Modulus of the Function $\log|{\cal U}(s)|$}
We commence with the proposition enunciated above that the magnitude of the analytic function ${\cal U}(s)$ is strictly larger than unity to the left of the critical line, and, by the functional equation (\ref{uthm1}), strictly less than unity in the region to the right of it, provided that $t>2.94334$.
We first present a proof for a related proposition, using the language of electrostatics in two dimensions. This of course is relevant to the case of analytic functions of a complex variable, given the correspondence between the governing equations of the two fields (Morse and Feshbach,1981, Roos, 1969, Bremermann, 1965).
\begin{theorem} Suppose $u(\mathbf{x})$ is the logarithmic electrostatic potential for a two dimensional problem in which all positive charges  lie in $x<0$ and all negative charges are mirror images of the
positive charges lying in $x<0$, so that $x=0$ is an equipotential line with $u(0,y)=0$ for all $y$.
Then $u(x,y)>0$ if $x<0$ and $u(x,y)<0$ if $x>0$.
\label{thelectrostatics}
\end{theorem}
\begin{proof}
Given an (inhomogeneous) Laplace equation of the form
\begin{equation}
\nabla^2 u(\mathbf{x}) =2\pi  \sum_p m_p[ \delta (\mathbf{x} -\mathbf{x}_p)-\delta (\mathbf{x} +\mathbf{x}_p)]
\end{equation}
where $\mathbf{x} = (x,y)$, $m_p$  denotes a positive coefficient, and all the points $\mathbf{x}_p$ lie in the $x>0$  half-plane,
we write the solution as
\begin{equation}
u(\mathbf{x}) = \sum_p m_p[V(| \mathbf{x} - \mathbf{x}_p|) - V(|\mathbf{x} + \mathbf{x}_p|)]
\end{equation}
where
$$
V(R) =  \ln R= \ln |\mathbf{x} - \mathbf{x}_p| ~.
$$
Now, $V(R)$ is a monotonically increasing function ($V(a)>V(b)$ for any $a>b$), so for all points
in the right half-plane (which is where $|\mathbf{x} + \mathbf{x}_p|<|\mathbf{x} - \mathbf{x}_p|$) we have
$V(|\mathbf{x} - \mathbf{x}_p|) < V(|\mathbf{x} + \mathbf{x}_p|)$. So for all $\mathbf{x}_p$,
$\mathbf{x}$ lying in the right half-plane we have
$$
V(|\mathbf{x} - \mathbf{x}_p|) - V(|\mathbf{x} + \mathbf{x}_p|) < 0
$$
and therefore
\begin{equation}
u(\mathbf{x}) = \sum_p m_p[V(|\mathbf{x} - \mathbf{x}_p|) - V(|\mathbf{x} + \mathbf{x}_p|)] < 0 ~~~~\forall x > 0 ~~.
\end{equation}
\end{proof}

This result is of course physically obvious: all the positive charges lie to the left of an equipotential line, and the negative charges to its right. Therefore the potential must increase as one moves to the left of the equipotential line and towards the positive charges, and must always exceed the potential at the mirror point to the right of the critical line.

The adaptation of this result to the case of $\log |{\cal U}(s)|$ has to take into account the exceptional pole and zero lying on the axis $t=0$. These  give rise to the atypical region evident in Fig. \ref{fig03t} for ${\cal V}(s)$, which, using the discussion in Theorem 3.2, carries over to ${\cal V}(s)$.
We also note that the positive coefficients
$m_p$ take into account the possible existence of poles  and zeros of multiplicity larger than unity 
(with however $m_p=1$ for all  poles and zeros in the range of $t$ which has been numerically investigated). 

\begin{corollary}
The function $|{\cal U}(s)|$ nowhere takes the value unity off the critical line in the finite part of the
complex plane of $s$, if $t>2.94334$.
\label{corabsu}
\end{corollary}
\begin{proof}
We apply the expansion of Theorem \ref{thelectrostatics} to the logarithmic modulus of 
${\cal U}(s)$, taking into account its zero and pole on the axis $\Im (t)=0$:
\begin{equation}
\log |{\cal U}(s)|=-\Re\left[ \log \left( \frac{s-1}{s}\right)\right]-
\sum_{p=1}^\infty m_p \Re\left[\log\left(\frac{s-s_p}{s-1/2-s_p} \right)+
\log\left(\frac{s-\bar{s_p}}{s-1/2-\bar{s_p}} \right)
\right]
\label{corabsu1}
\end{equation}
Here the $s_p$ denotes the $p$ th pole of  ${\cal U}(s)$, located at the point $s=\sigma_p+i t_p$,
where $\sigma_p=1/4$ and $m_p=1$  for all poles encountered numerically, and $2 t_p$ is the imaginary part of the $p$th zero of $\zeta(s)$. Note that the zeros of  ${\cal U}(s)$ in the upper half plane and its zeros and poles in the lower half plane have been included explicitly in equation (\ref{corabsu1}).
Note also that the representation (\ref{corabsu1}) satisfies $|{\cal U}(1/2+ it)|=0$. (An alternative route to the equation (\ref{corabsu1}) is to use the product representations
incorporating the zeros
of the numerator and denominator in the definition of ${\cal U}(s)$, (\ref{udef2}), and to take the logarithm of that quotient.)

We note that the result of Theorem \ref{thelectrostatics} applies to the sum over $p$ in equation (\ref{corabsu1}), or any subset of it. We pick out the term $p=1$ to counterbalance the exceptional term, and apply Theorem \ref{thelectrostatics}  to the sum over $p$ from two to infinity. We then examine the value of the first term in (\ref{corabsu1}) combined with the term $p=1$, or
the modulus of
\begin{equation}
\left( \frac{s-1}{s}\right) \left(\frac{s-s_1}{s-1/2-s_1} \right)
\left( \frac{s-\bar{s_1}}{s-1/2-\bar{s_1}} \right) .
\label{corabsu2}
\end{equation}
The modulus of this expression is unity when
\begin{equation}
(1-2\sigma)(1+16 t_1^2)(1-16\sigma+16\sigma^2-48 t^2+16 t_1^2)=0,
\label{corabsu3}
\end{equation}
which occurs when $\sigma=1/2$ or
\begin{equation}
t=\frac{1}{\sqrt{3}} \sqrt {t_1^2-\frac{3}{16}+\left(\sigma-\frac{1}{2}\right)^2}.
\label{corabsu4}
\end{equation}
Replacing $t_1$ by its value of approximately 7.06736, this gives a value of 4.07268 for $\sigma=1/2$, which is a minimum with respect to variation of $t$ from (\ref{corabsu4})
with $\sigma$. (In fact, the equation (\ref{corabsu4}) gives a value for $t$ which varies only from 4.0727 to 4.1134 as $t$ ranges from -0.5 to 1.5.) 

The simple analytic estimate of 4.07268  just obtained is replaced by the numerical value 2.94334 when the terms $p=2,3,\ldots$ are taken into account in equation (\ref{corabsu1}).
\end{proof}

The expansion (\ref{corabsu1}) is absolutely convergent for any finite value of $t$. Indeed,  we divide the infinite sum into an infinite number of terms with $t_p>>t$, and a finite sum. The infinite series then converges as a sum over $1/t_p^2$, given that $\sigma_p$ lies in the open interval $(1/2,1)$. Any density of zeros function less strong than the first power of $t$ will leave the expansion absolutely convergent. The accuracy of the representation (\ref{corabsu1}) can easily be verified numerically by direct summation, given a table of zeros of $\zeta (s)$. (Further comments expanding these remarks may be found in  Appendix 1.)
\begin{corollary}
All zeros of ${\cal T}_+(s)$ and ${\cal T}_-(s)$ lie on the critical line, and are of order unity.
\label{corol2}
\end{corollary}
\begin{proof}
From equation (\ref{tdef}), all zeros or poles of ${\cal V}(s)$ must correspond to points where  ${\cal U}(s)=\pm 1$, i.e. where $|{\cal U}(s)|=1$. From Corollary \ref{corabsu}, such points can only occur on the line $\sigma=1/2$. This proves the first proposition, and extends the result of Taylor (1945) from ${\cal T}_-(s)$ to ${\cal T}_+(s)$.

With regard to the order of the zeros or poles of ${\cal V}(s)$, we note that zeros or poles of order greater than unity would require more than at least a second line of constant phase emanating
from them  (in addition to the critical line) along which ${\cal V}(s)$ is pure imaginary. Such lines would have to be lines of unit amplitude for ${\cal U}(s)$,
and so in fact are limited to the critical line by Corollary \ref{corabsu}.
\end{proof}

\begin{corollary}
${\cal U}'(s)$  and ${\cal V}'(s)$ have no zeros on the critical line in $t>2.94334$.
\label{corol1}
\end{corollary}
\begin{proof}
We use the method of Corollary \ref{corabsu}, and the expansion (\ref{corabsu1}) to evaluate
the partial derivative with respect to $\sigma$ of $\log |{\cal U}(s)|$:
\begin{equation}
\left. \frac{\partial \log |{\cal U}(s)|}{\partial \sigma}\right|_{\sigma=1/2}=
\frac{1}{t^2+1/4}+\sum_{p=1}^\infty m_p \frac{8(2\sigma_p-1)(4 t^2+4 t_p^2+4 \sigma_p (\sigma_p-1)+1)}{[(1-2\sigma_p)^2+4 (t-t_p)^2] [(1-2\sigma_p)^2+4 (t+t_p)^2]}.
\label{sigder1}
\end{equation}
The first term is exceptional, and gives a positive contribution, whereas all terms from the sum over $p$ give a negative contribution, since $\sigma_p<1/2$.  The negative contributions outweigh the
exceptional term as soon as $t$ exceeds 2.94334: see Fig. \ref{figderiv}. From this value on, we
have that the partial derivative of $\log |{\cal U}(s)|$ with respect to $\sigma$ on the critical line is
strictly negative. This guarantees that ${\cal U}'(s)$  is never zero in $t>2.94334$, while the corresponding result for ${\cal V}'(s)$ follows from equation (\ref{uthm4a}).
\end{proof}

\begin{corollary}
$F_1(s)$ has all its zeros and poles on the critical line, and all these are of first order, provided   $t>2.94334$.\label{corf1}
\end{corollary}
\begin{proof}
The zeros and poles of $F_1(s)$  occur where ${\cal U}(s)=i, -i$ respectively- i.e., in both cases $|{\cal U}(s)|=1$. Thus, both poles and zeros lie on $\sigma=1/2$.  From equation (\ref{f1-3}) and Corollary \ref{corol1}, we see that
$F_1'(s)$ is never zero on the critical line if $t>2.94334$, so the poles and zeros are of order unity. (The statement concerning the zeros of $F_1(s)$ being first order is obvious.
The statement concerning poles relies on the connection between $F_1(s)$ and ${\cal V}(s)$, and the knowledge that ${\cal V}(s)$ and its derivative remain finite at the pole of $F_1(s)$.)
\end{proof}

\begin{figure}[h]
\includegraphics[width=3.0in]{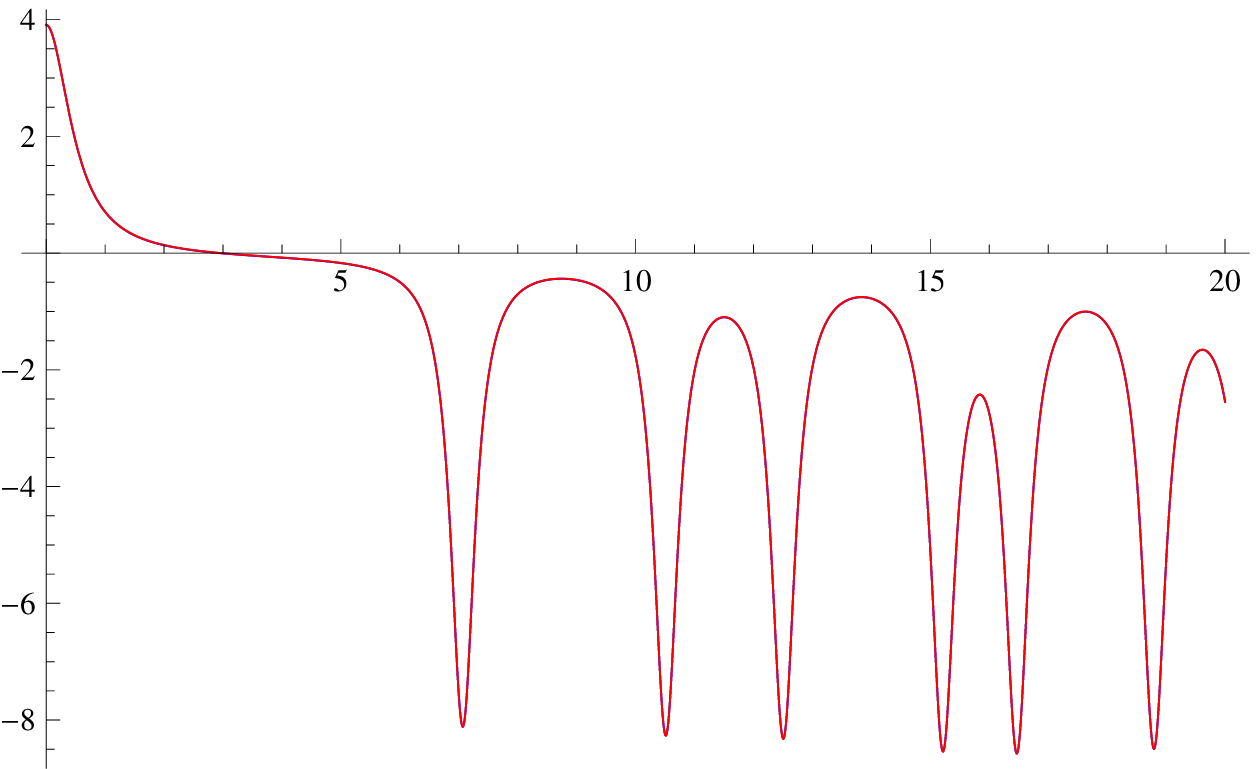}~~\includegraphics[width=3.0in]{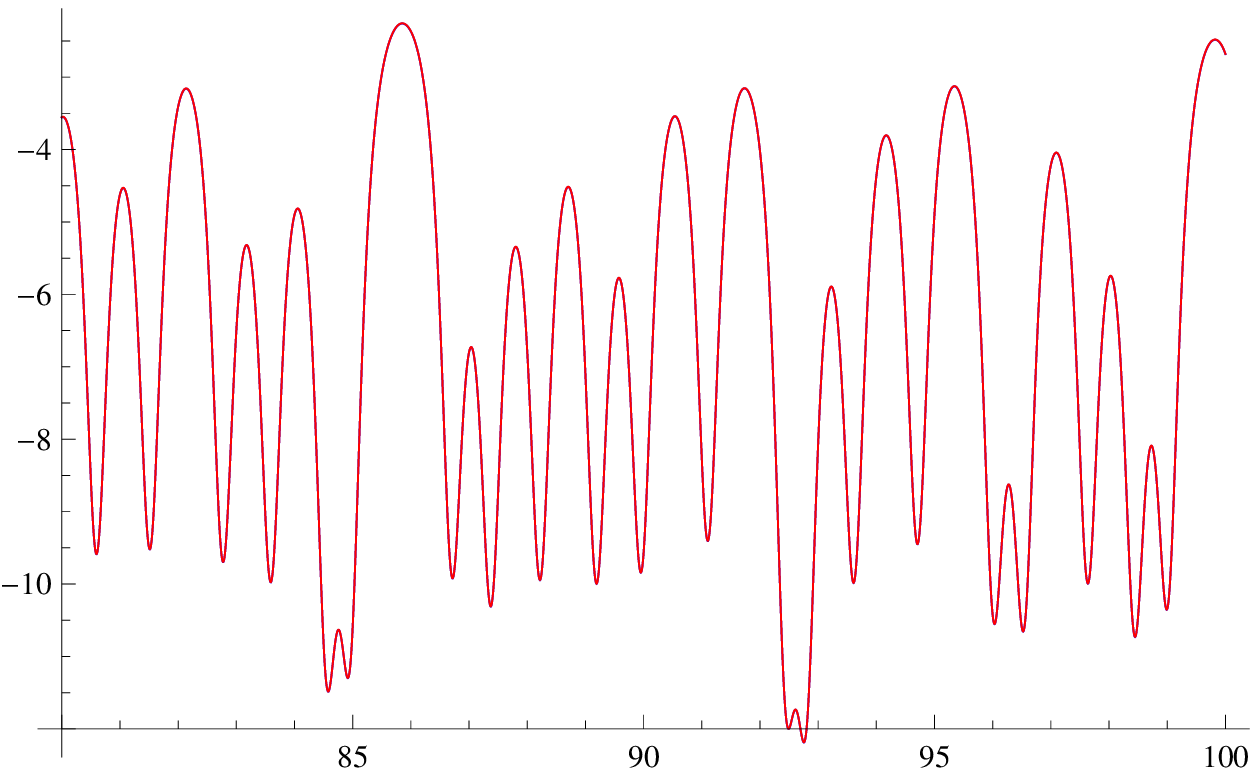}
\caption{\label{figderiv}  Plots of the  partial derivative of $\log |{\cal U}(s)|$ with respect to $\sigma$ on the critical line. The plots of the numerically-differentiated function and the result from equation (\ref{sigder1}), summed over $p$ from 1 to 1000, coincide to graphical accuracy.  }\end{figure}

The derivative plots in Fig. \ref{figderiv} show a function which has negative peaks whenever 
$t$ passes through a value $t_p$. The maximum value between the negative peaks is always negative once $t$ exceeds 2.94334, and decreases slowly as $t$ increases.

Note that, from (\ref{ecor1}), we also have that $F_1'(1/2+i t)$ and $\Psi(t)=\partial \arg \xi_1(1+2 i t)/\partial t$ are never zero in $t>2.94334$. From this, it follows that
$\arg {\cal U}(1/2+i t)$ decreases as $t$ increases, while $\Im {\cal V}(1/2+ i t)$ increases, as does $\Re F_1(1/2+i t)$.

\begin{corollary}
Zeros and poles of ${\cal V}(s)$ strictly alternate on the critical line, and so ${\cal T}_+(s)$
and ${\cal T}_-(s)$ have the same distribution function for zeros.
\label{corol3}
\end{corollary}
\begin{proof}
Since ${\cal V}'(s)\neq 0$ on the critical line, $\log[|{\cal V}(1/2+i t)|]$ is a monotonic function of $t$.  Starting at a pole with the value $+\infty$, it decreases, passing through each real value once before arriving at a zero, with the value $-\infty$.  Starting at a zero, it increases monotonically until it reaches a pole. (Put in another way, after its zero derivative point at $(1/2, 2.94334)$, 
$\Im[{\cal V}(1/2+i t)]$ increases monotonically, except at its first-order poles, where it changes
discontinuously from $+\infty$ to $-\infty$.)

Given this result, it follows that the distribution function $N({\cal T}_+)$ which gives the number of zeros  of  ${\cal T}_+(s)$
on the critical line up to the ordinate $t$ must agree with the corresponding function $N({\cal T}_-)$ for ${\cal T}_-(s)$ in all terms which diverge as $t\rightarrow \infty$.
\end{proof}

\begin{corollary} The real function $\arg[\xi_1(1+2 i t)]$ is monotonic increasing in $t>2.94334$, apart from discontinuous drops from $\pi$ to $-\pi$, which occur at the same values of $t$ as the poles of ${\cal V}(1/2+i t)$. The distribution function of the values where $\arg[\xi_1(1+2 i t)]=n\pi$ for positive integers $n$ is the same in all terms which go to infinity with $t$ as those of the zeros of ${\cal T}_+(s)$ and 
${\cal T}_-(s)$.
\label{newcor}
\end{corollary}
\begin{proof}
From equation (\ref{ecor1}) the zeros of ${\cal T}_+(s)$ occur when $t$ satisfies $\arg[\xi_1(1+2 i t)]=(n-1/2)\pi$ for integer $n$, while those of  ${\cal T}_-(s)$ occur when $t$ satisfies $\arg[\xi_1(1+2 i t)]=n\pi$. These are first-order, and so the phase change  of $\arg[\xi_1(1+2 i t)]$ from $\pi$ to $-\pi$ occurs at the poles of ${\cal V}(1/2+i t)$. The monotonicity of $\arg[\xi_1(1+2 i t)]$ is a consequence of that of
$\Im[{\cal V}(1/2+i t)]$- see Corollary \ref{extracor}.
\end{proof}

This result is related to that in Theorem 11.1 in Titchmarsh\& Heath-Brown (1987), which is that $\zeta(s)$ takes every value, with one possible exception, an infinity of times in any strip $1-\delta<\sigma\leq1+\delta$. Here the function in question is $\arg[\xi_1(1+2 i t)]$, which takes all real values an infinity of times, with no exception. Note also that the same comments apply to $\arg[\xi_1(2 i t)]$, except that it is monotonic decreasing rather than increasing.

\begin{theorem}
Consider an interval  ${\cal I}$ on the critical line which begins with a pole   $P_1$ of ${\cal V}(s)$, continues with a zero $Z_1$ and terminates with a pole $P_2$. 
This maps onto a circle $|{\cal U}(s)|=1$, which contains a single  zero of ${\cal U}(s)$, which corresponds to a phase jump of $\pi$. This establishes a direct correspondence between zeros of ${\cal U}(s)$
and intervals on the critical line between poles of ${\cal V}(s)$.
\label{thmconn}
\end{theorem}
\begin{proof}
By the argument of Corollary \ref{corol3} we know that $\Im [{\cal V}(s)]$ is negative for $t$ above $P_1$. It will become less and less negative as $t$ increases, and will continue to increase
until it reaches the next pole ($P_2$).  It will successively pass through the first fixed point ($F_1$,
where ${\cal V}(s)=-i$), the first-order zero ($Z_1$),
and the second fixed point ($F_2$,
where ${\cal V}(s)=+i$),  before reaching $P_2$. $F_1$ and $F_2$ are the only points encountered where $|{\cal V}(s)|=1$.

Along this trajectory, ${\cal V}(s)$ has traced out the entire imaginary axis. From the properties
of the mapping (\ref{uthm4}), this shows that ${\cal U}(s)$ has simultaneously traced out the
boundary of the unit circle in its values, with its phase varying through a range of $2 \pi$ in monotonic fashion. This shows that the centre of the circle in the ${\cal U}$ plane is
a  zero of ${\cal U}(s)$. The phase variation is consistent with a first-order zero, or a second order zero shared between two Riemann sheets, or a  zero of order $n$ shared between $n$ Riemann sheets. (See Appendix 2 for a numerical example.)

The mapping between ${\cal V}(s)$ and ${\cal U}(s)$ is one-to-one, so each interval between poles of ${\cal V}(s)$ on the critical line  maps directly onto a zero of ${\cal U}(s)$ .
\end{proof}

\begin{comment}
The mapping from intervals on the critical line between poles of ${\cal V}(s)$ to zeros of ${\cal U}(s)$ is unique. As we shall see in Appendix 2, this does not necessarily mean that each zero of   ${\cal U}(s)$ is associated with only one such interval on the critical line.
\end{comment}

We can also extend the discussion of Theorem \ref{thmconn} by  constructing the contour $|{\cal V}(s)|=1$. This is in the plane of complex ${\cal V}$
values a circle, with centre at  $Z_1$ and with $F_1$ and $F_2$ lying on its boundary. In the complex $s$ plane, it is a closed contour around which the phase $\arg [{\cal V}(s)]$ varies
by $2\pi$, in keeping with there being a single zero enclosed by the contour (and no singularity).
Along the  contour, the phase of ${\cal U}(s)$ is either $+\pi/2$ or $-\pi/2$, and   the entire imaginary axis must be traced out once. This means that on this contour we must encounter a  first-order zero and a pole. A contour of phase $\arg[{\cal U}(s)]=\pi$ must pass through the zero and the pole; this must be symmetric under reflection in the critical line. This is also a contour of phase
$\arg [{\cal V}(s)]=\pi$ where $|{\cal U}(s)|>1$ ($\sigma<1/2$) and $\arg [{\cal V}(s)]=0$ where $|{\cal U}(s)|<1$ ($\sigma>1/2$).

\begin{corollary}
If the Riemann hypothesis holds, then the  distribution functions of zeros of ${\cal T}_+(s)$ and ${\cal T}_-(s)$ must agree with that for
$\zeta(2 s-1)$ in all terms which do not remain finite as $t\rightarrow \infty$.
\label{equidist}
\end{corollary}
\begin{proof}
Given the assumption that the Riemann hypothesis holds, then all zeros and poles of ${\cal U}(s)$ lie respectively on the lines $\sigma=3/4$ and $\sigma=1/4$. Each zero and pole lies at the intersection of contours of constant phase $\arg {\cal U}(s)=\pi/2$ and $\arg {\cal U}(s)=-\pi/2$, i.e., the contours of constant amplitude $|{\cal V}(s)|=1$. There are two such contours
 surrounding each zero of ${\cal V}(s)$, so there is a one-to-one mapping between zeros of ${\cal U}(s)$ and of ${\cal V}(s)$. We can similarly argue that each pole and zero  of
 ${\cal U}(s)$ are linked by a line of constant phase zero (which passes through the zero of  ${\cal V}(s)$) and a line of constant phase $\pi$ (which passes through the pole of  ${\cal V}(s)$)-- a second one-to-one mapping.
\end{proof}

The work of R.C. McPhedran  and C.G. Poulton on this project has been supported by the Australian Research Council's Discovery Grants Scheme, while the former also acknowledges the financial support of the European Community's Seventh Framework
Programme under contract number PIAPP-GA-284544-PARM-2. M.
The referees of this paper are thanked for insightful comments.

\newpage
\section*{Appendix 1: Convergence of the Logarithmic Potential Expansion}
We consider the convergence of the series for $\log|{\cal U}(s)|$ in equation (\ref{corabsu1}).
We expand the summand of the series for all terms in which $t_p>>t$, with $t>>1$. Expanding the
series with these assumptions, the $p$th term is to leading order
\begin{equation}
T_p (\sigma, t)=m_p \Re\left[\frac{(1-2\sigma_p) (1-2\sigma-2 i t)}{t_p^2}\right] +O\left( \frac{t^2}{t_p^4}\right)
=m_p \left[\frac{(1-2\sigma_p) (1-2\sigma)}{t_p^2}\right] +O\left( \frac{t^2}{t_p^4}\right).
\label{app1}
\end{equation}
We may then apply the integral test to assure the convergence of the series in equation (\ref{corabsu1}). To do this, we use the result( Titchmarsh \& Heath-Brown, 1987)  that the leading term in the density of zeros in the critical strip $0<\sigma<1$ of $\zeta (s)$ is   $\log (t)/(2 \pi)$, so that the convergence is evident.

To investigate numerically the representation of $\log|{\cal U}(s)|$ in equation (\ref{corabsu1}), we consider the application of the Euler-Maclaurin formula (Apostol, 1999) to the latter. The formula  of order $q$ is
\begin{equation}
\sum_{p=0}^\infty f(p)=\int_0^\infty f(x) dx -B_1 f(0)-\sum_{k=1}^q \frac{B_{2 k}}{(2 k)!} f^{(2 k-1)}(0) +R,
\label{app2}
\end{equation}
where $R$ denotes the remainder term of order $q$,  the $B_{2 k}$ are Bernoulli numbers, and we have assumed the function $f(x)$ and its derivatives up to order
$2 q-1$ tend to zero as $x\rightarrow \infty$. We apply this by summing directly in  (\ref{corabsu1}) up to $p=L$, and using the Euler-Maclaurin formula to estimate the remaining infinite sum. This application is not conceptually straightforward, since the zeros $t_p$ do not have a known functional dependence on $p$, and they are not even spaced, as is implicit in the derivation of the Euler-Maclaurin formula. We then write the integral as
\begin{equation}
\int_{L+1}^\infty T_p (\sigma, t) \frac{\log(t_p/\pi)}{\pi} dt_p= \frac{(1-2\sigma)}{2(L+1)\pi} \left[ 1+\log\left(\frac{L+1}{\pi}\right)\right],
\label{app3}
\end{equation}
where we have used the assumptions $m_p=1$ and $\sigma_p=1/4$, together with the known terms of the density function for zeros of $\zeta (s)$ in the critical strip, in  the
establishment of a functional form of the integrand. The derivative terms in equation (\ref{app2}) need to be evaluated using the same density function for $t_p$ values.

As an example of the accuracy of the resulting numerical estimate, with $L=1000$ and $q=0$, so only the first two terms in the right-hand side of (\ref{app2}) are used,
the difference between the highly accurate value for $\log{\cal U}(0.4+ i t)$ from Mathematica and the numerical estimate from the  Euler-Maclaurin formula is bounded from above by 0.000078 for $t$ in the range of $t$ between zero and 100. The values of $\log{\cal U}(0.4+ i t)$ are of order unity for $t$ in this range.

\section*{Appendix 2: Remarks on the Connection between ${\cal U}(s)$ and ${\cal V}(s)$}
 Let us consider once more the connection between the three functions ${\cal U}(s)$,  ${\cal V}(s)$ and $F_1(s)$ embodied in equation (\ref{uthm5}). We can express
 both ${\cal U}(s)$ and ${\cal V}(s)$ in terms of $F_1(s)$:
 \begin{equation}
 {\cal U}(s)=\frac{i (i+F_1(s))}{(i-F_1(s)},~~ {\cal V}(s)=\frac{i (1+F_1(s))}{(1-F_1(s)}.
 \label{rh1}
 \end{equation}
 We can thus connect ${\cal U}(s)$ and ${\cal V}(s)$ in an invariant way to $F_1(s)$ if we plot the first two as a function of the real and imaginary parts of the third. The results of this are shown in Fig. \ref{figf1plane}. They illustrate explicitly the result that constant amplitude lines and constant phase lines of ${\cal U}(s)$ are obtained from those of
 ${\cal V}(s)$ by a rotation through $90^\circ$ in the complex $F_1$ plane.
 
 \begin{figure}[h]
\includegraphics[width=3.0in]{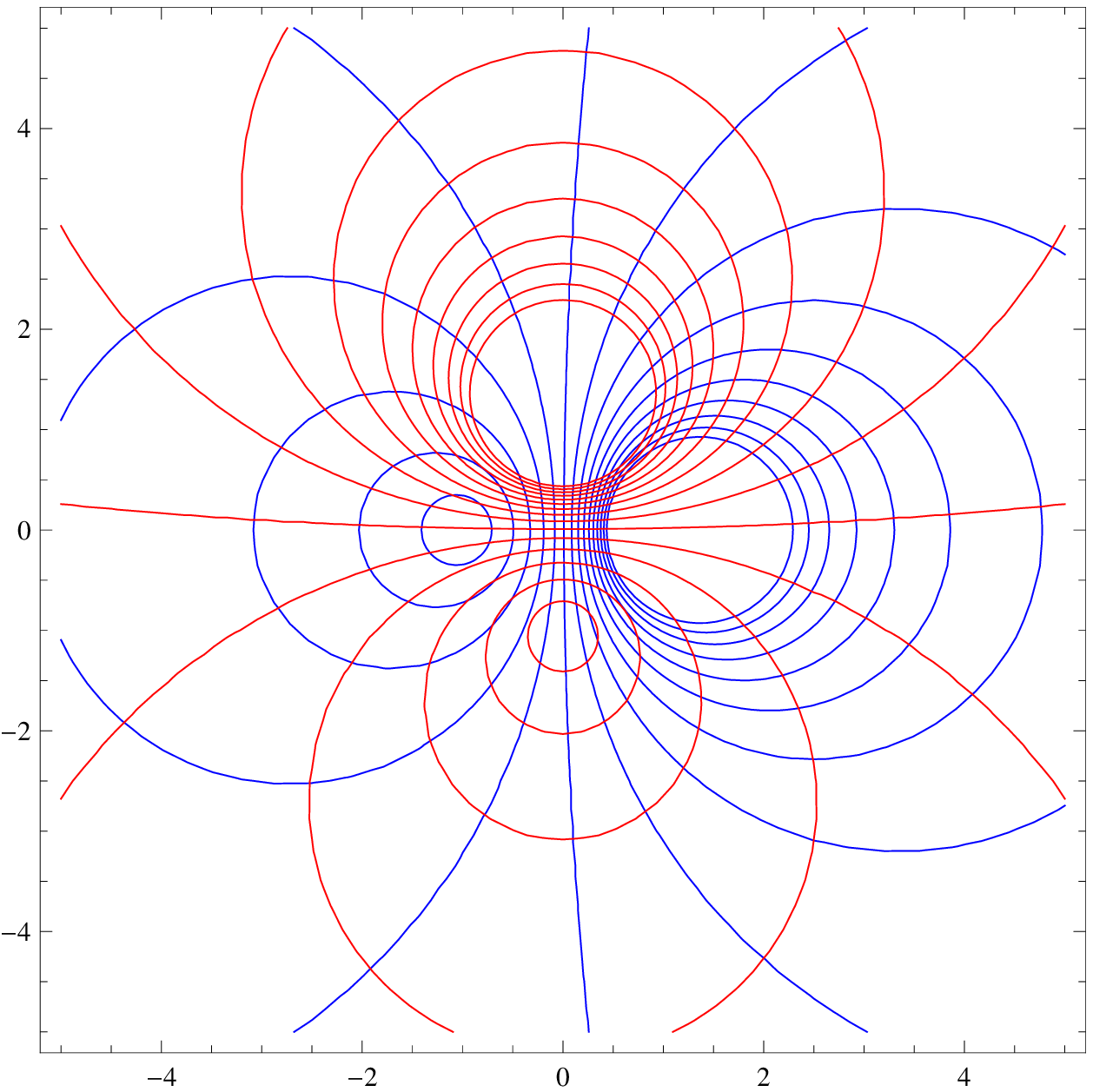}~~\includegraphics[width=3.0in]{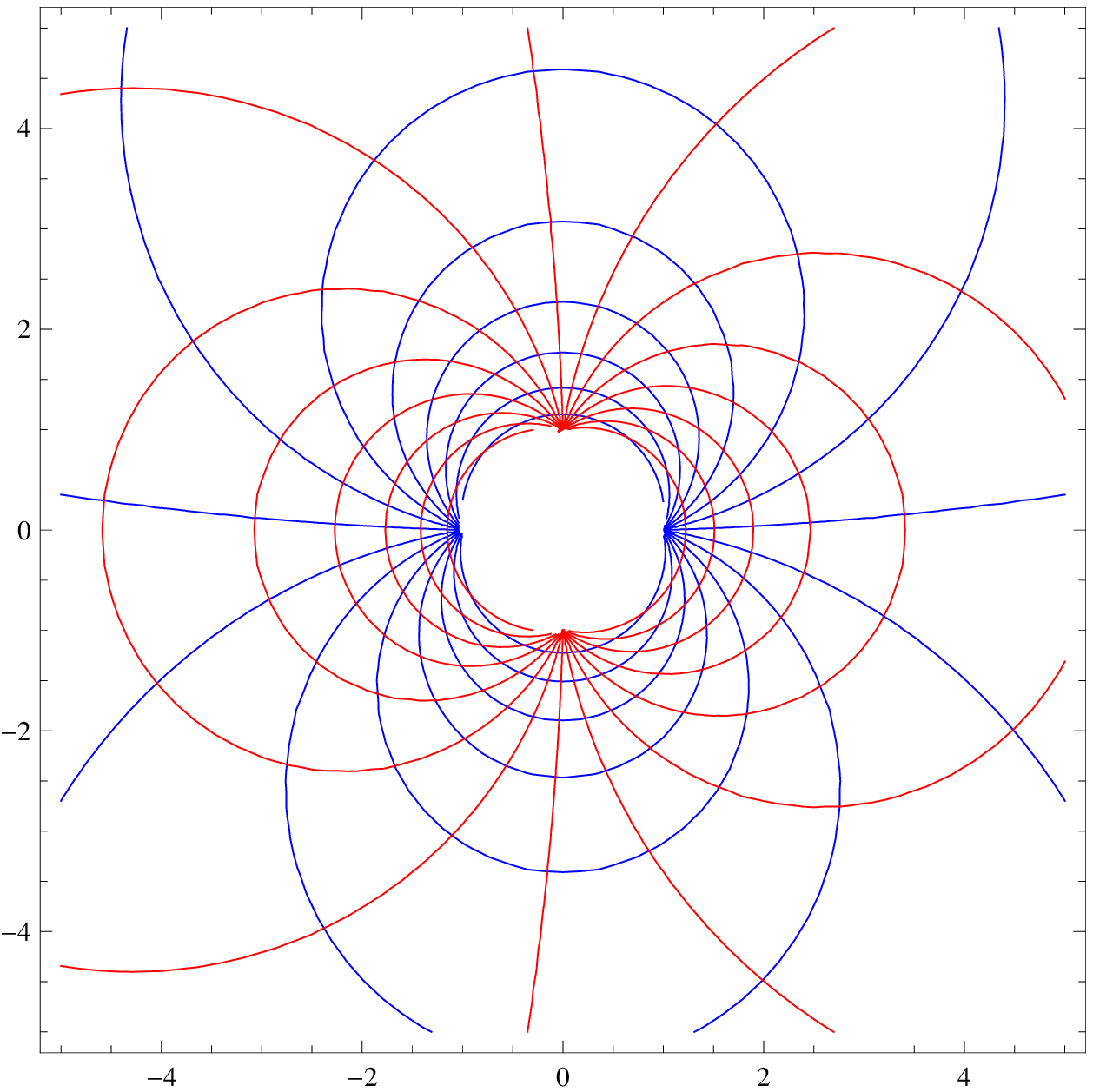}
\includegraphics[width=3.0in]{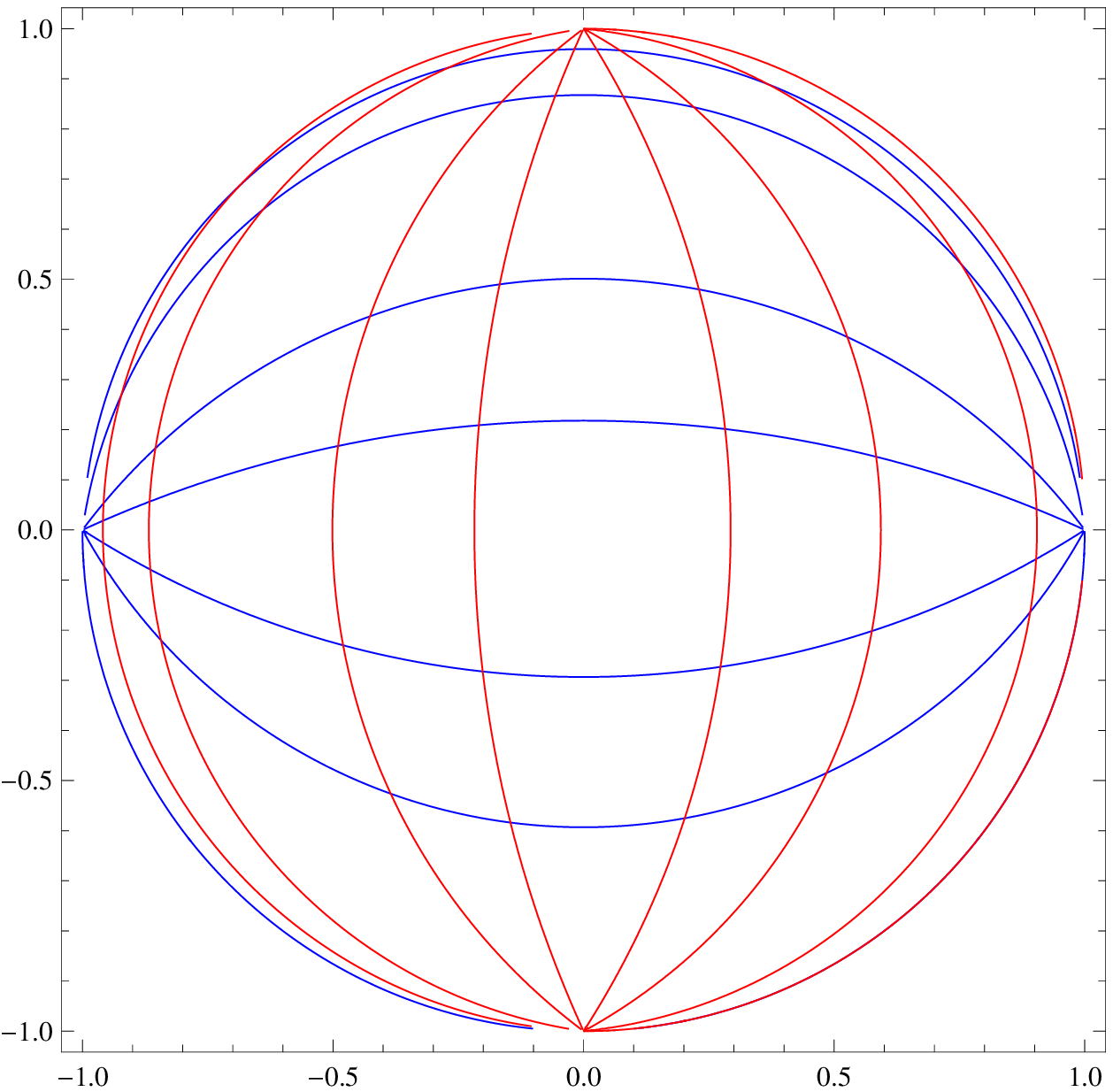}
\caption{\label{figf1plane}  Plots of lines of constant amplitude (top left) and constant phase (top right) of ${\cal V}(s)$ (blue lines) and of   ${\cal U}(s)$ (red lines)  as a function of the
real and imaginary parts of ${\cal F}_1(s)$. Bottom:Plots of constant phase for $|F_1|<1$. }\end{figure}

\begin{theorem}
If all zeros of ${\cal U}(s)$ are of first order, each corresponds to one Riemann sheet of $F_1(s)$ and one Riemann sheet of ${\cal V}(s)$.
\label{rhproof}
\end{theorem}
\begin{proof}
We have proved that in each interval on the critical line between successive poles of ${\cal V}(s)$ there is one zero of ${\cal V}(s)$, and one pole of $F_1(s)$ followed by a zero
of $F_1(s)$. All the poles and zeros mentioned are first-order. The fact that $F_1(s)$ has its poles and zeros on the critical line, where its values are real and it increases monotonically, shows that in between successive poles (where $s$ takes the values $s_1$ and $s_2$) it takes on values which map out the real axis. By analytic continuation, its values will then fill the complex $F_1$ plane for a set of values of $s$ built around the critical line element running from $s_1$ to $s_2$. The set  of  $s$ values and corresponding $F_1(s)$ values constitutes a Riemann sheet of $F_1(s)$ (Knopp, 1947). The corresponding contours of amplitude and phase of ${\cal V}(s)$ and ${\cal U}(s)$ are shown
in Fig. \ref{figf1plane}; in each case they include precisely one first order zero of ${\cal V}(s)$ and ${\cal U}(s)$.

The set of zeros of ${\cal U}(s)$ so mapped out as $s$ ranges over the entire critical line constitutes the entire set of zeros of this analytic function.  The corresponding statement for poles follows from
the functional equation (\ref{uthm1}) for ${\cal U}(s)$. If ${\cal U}(s)$ has zeros off the line $\sigma=3/4$, these will occur in pairs placed symmetrically about $\sigma=3/4$, each corresponding to a different sheet of $F_1(s)$ and of ${\cal V}(s)$, and thus to a different interval of $t$ on the critical line.
\end{proof}

The arguments of Section 4 are quite general, and apply even if certain properties of the function ${\cal U}(s)$ are altered. For example, we can choose a particular zero/pole pair, say the
$N$th zero/pole in $t=t_N>0$ and artificially double their order in a new function $\tilde{\cal U}_N(s)$:
\begin{equation}
\tilde{\cal U}_N(s)={\cal U}(s) \frac{[s-(3/4+i t_N)][s-(3/4-i t_N)]}{[s-(1/4+i t_N)][s-(1/4-i t_N)]},
\label{tildecaludef}
\end{equation}
and define
\begin{equation}
\tilde {\cal V}_N(s)=\frac{1+\tilde{\cal U}_N(s)}{1-\tilde{\cal U}_N(s)}.
 \label{tildeNcalvdef}
\end{equation}
Alternatively, we can double the order of every zero and pole:
\begin{equation}
\tilde{\cal U}(s)={\cal U}^2(s),~~\tilde {\cal V}(s)=\frac{1+\tilde{\cal U}(s)}{1-\tilde{\cal U}(s)}.
\label{tildecalvdef}
\end{equation}
 \begin{figure}[h]
\includegraphics[width=3.0in]{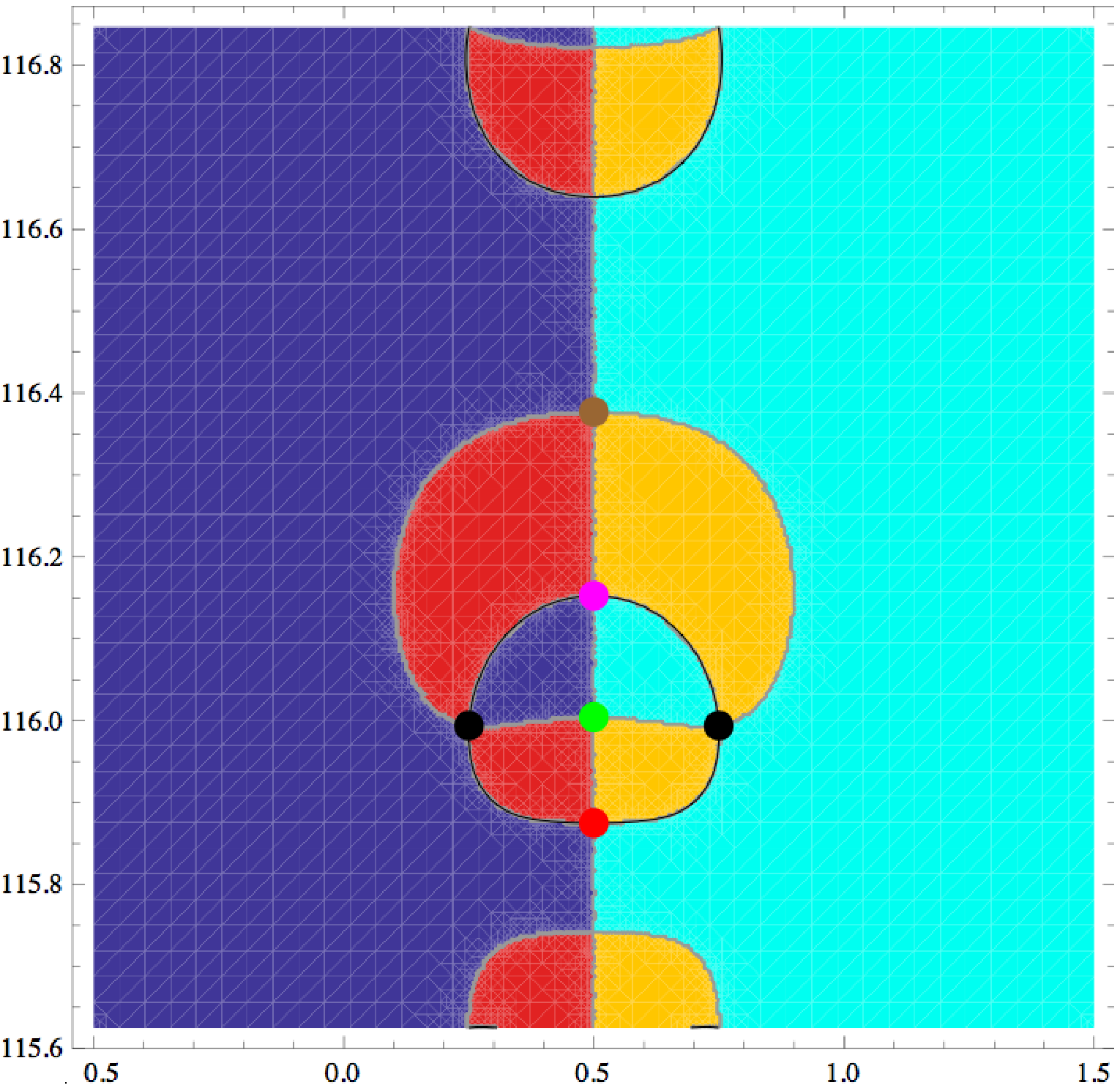}~~\includegraphics[width=3.0in]{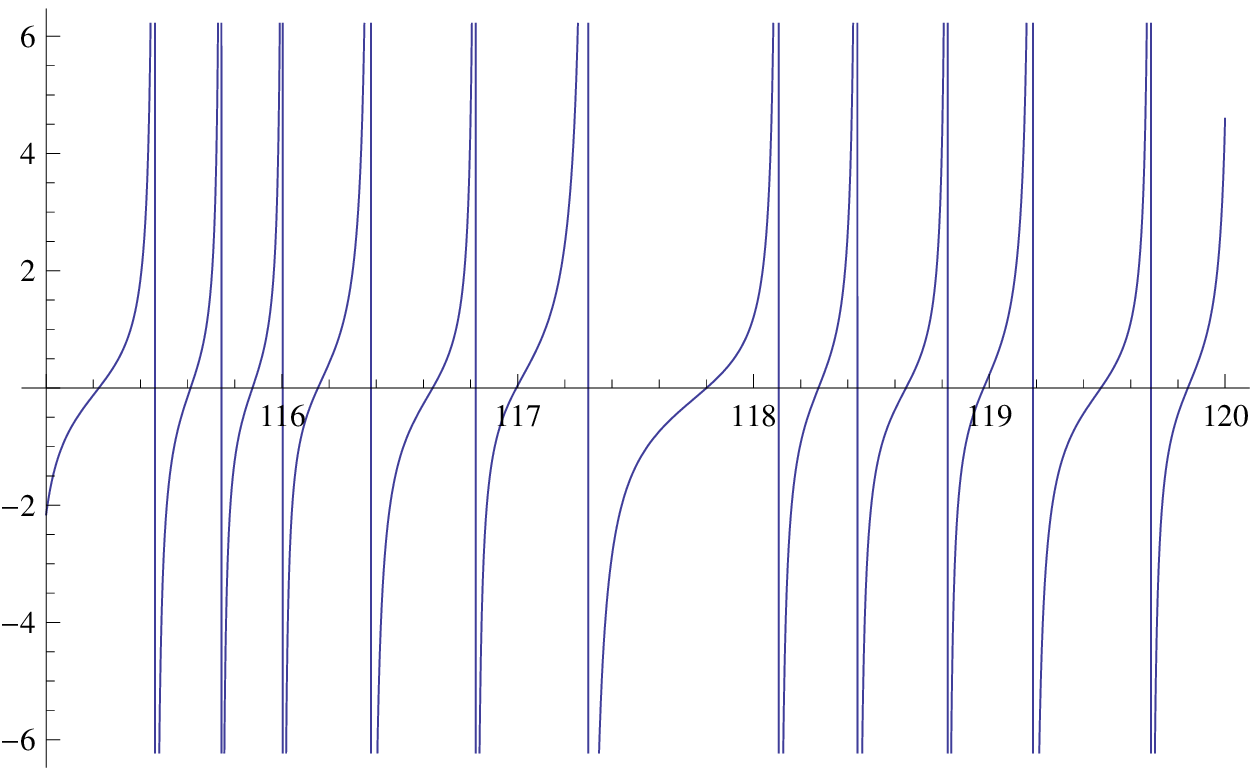}
\caption{\label{cex98}  Left: plot of the phase of $\tilde {\cal V}(s)$ in the region near the 98th zero and pole, which, with all other zeros and poles, have artificially been made second-order.
Right: plot of  $\Im \tilde {\cal V}(s)$ in the region around  $t=t_{98}$. }\end{figure}
The results of such a numerical experiment are shown in Fig. \ref{cex98}. On the left, we see two Riemann sheets of $\tilde {\cal V}(s)$ meeting at the second-order pole
and zero. On the right, we see that $\Im \tilde {\cal V}(s)$ is still monotonic increasing on the critical line, with first-order zeros and poles occurring with twice the density
as compared with the case of ${\cal V}(s)$.

As a second example, we can take the $N$th zero/pole in $t=t_N>0$ and artificially render them into a pair of zeros/poles symmetrically located around the lines
$\sigma=3/4$, $\sigma=1/4$ respectively:
\begin{equation}
\hat{\cal U}(s)={\cal U}(s) \frac{[s-(3/4-\delta+i t_N)][s-(3/4+\delta-i t_N)]}{[s-(1/4-\delta+i t_N)][s-(1/4+\delta-i t_N)]}\frac{[s-(1/4+i t_N)]}{[s-(3/4+i t_N)]},
\label{hatcaludef}
\end{equation}
and
\begin{equation}
\hat {\cal V}(s)=\frac{1+\hat{\cal U}(s)}{1-\hat{\cal U}(s)}.
 \label{hatcalvdef}
\end{equation}
 \begin{figure}[h]
\includegraphics[width=3.0in]{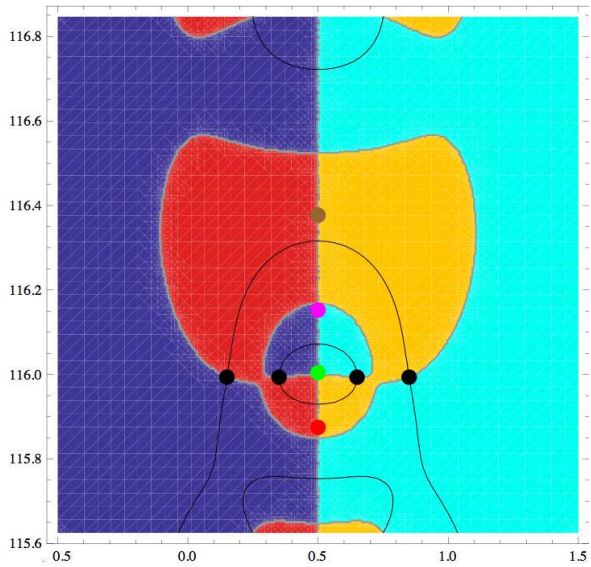}
\caption{\label{chatex98}  Plot of the phase of $\hat {\cal V}(s)$ in the region near the 98th zero and pole, which have artificially been split into a pair of zeros and a pair of poles,
denoted by the black dots.
}\end{figure}
Figure \ref{chatex98} shows that each of the split pairs corresponds to an interval on the critical line, with the inner pole and zero linked with the first interval and the outer pole and zero with the second.

\end{document}